%% file: Backscatter_Net_v16_noBlue.tex
\documentclass[12pt, draftclsnofoot, onecolumn]{IEEEtran}
 \usepackage{amsmath,amssymb}
 \usepackage{subfigure}
 \usepackage{graphicx,graphics,color,psfrag}
 \usepackage{cite,balance}
 \usepackage{caption}
 \allowdisplaybreaks
 \usepackage{algorithm}
 \usepackage{accents}
 \usepackage{amsthm}
 \usepackage{bm}
 \usepackage{algorithmic}
 \usepackage[english]{babel}
 \usepackage{multirow}
 \usepackage{enumerate}
 \usepackage{cases}
 \usepackage{stfloats}
 \usepackage{dsfont}
 \usepackage{color,soul}
 \usepackage{amsfonts}
 \usepackage{cite,graphicx,amsmath,amssymb}
 \usepackage{subfigure}
 \usepackage{fancyhdr}
 \usepackage{hhline}
 \usepackage{graphicx,graphics}
 \usepackage{array,color}
 \usepackage{amsmath}

\include{header}

\begin{document}

\addtolength{\textfloatsep}{-2pt}
\textheight=24cm

\title{\huge Wirelessly Powered Backscatter Communication Networks: Modeling, Coverage and Capacity}

\author{Kaifeng Han and Kaibin Huang  \thanks{\setlength{\baselineskip}{13pt} \noindent K. Han and K. Huang are with the Dept. of Electrical and Electronic Engineering at The  University of  Hong Kong, Hong Kong (Email: kfhan@eee.hku.hk, huangkb@eee.hku.hk). Part of this work has been accepted by IEEE Globecom 2016.}}
\maketitle
\begin{abstract} Future Internet-of-Things (IoT) will connect billions of small computing devices embedded in the environment and support their device-to-device (D2D) communication. Powering these massive number of embedded  devices is a key challenge of designing IoT since batteries increase the devices' form factors and battery  recharging/replacement is difficult. To tackle this challenge, we propose a novel network architecture that enables D2D communication between passive nodes by integrating wireless power transfer and backscatter communication, which is called a \emph{wirelessly powered backscatter communication} (WP-BackCom)  network. In this network, standalone \emph{power beacons} (PBs) are deployed for wirelessly powering nodes by beaming unmodulated carrier signals to targeted nodes. Provisioned  with a backscatter antenna, a node transmits data to an intended receiver by modulating and reflecting a fraction of a carrier signal. Such transmission by backscatter  consumes orders-of-magnitude less power than a traditional radio. Thereby, the dense deployment of low-complexity PBs with high transmission power can power a large-scale IoT. In this paper, a WP-BackCom network is modeled as a random Poisson cluster process in the horizontal plane where PBs are Poisson distributed and active ad-hoc pairs of backscatter communication nodes with fixed separation distances form random clusters centered at PBs. The backscatter nodes can harvest energy from and backscatter carrier signals transmitted by PBs. Furthermore, the transmission power of each node depends on the distance from the associated PB. Applying stochastic geometry, the network coverage probability and transmission capacity are derived and optimized as functions of backscatter parameters, including backscatter duty cycle, reflection coefficient, and the PB density. The effects of the parameters on network performance are characterized.
\end{abstract}

\section{Introduction}

\emph{Internet-of-Things} (IoT) is envisioned as a future network to connect billions of small computing devices embedded in the environment (e.g., walls and furniture) and implanted in bodies and enable their \emph{device-to-device} (D2D) wireless communication. Powering a massive number of such devices remains a key design challenge for IoT.  Batteries add to devices' weights and form factors and battery recharging/replacement increases  the maintenance cost if not infeasible.  To handle this challenge, we propose a novel network architecture that enables large-scale passive IoT deployment by seamlessly integration of wireless power transfer (WPT) \cite{Huang:CuttingLastWiress:2014, bi:PowerCommunication:2014} and low-power backscattering communication \cite{1boyer2014backscatterComm.},  called a \emph{wirelessly powered backscatter communication} (WP-BackCom) network. Specifically, \emph{power beacons} (PBs) that are  stations  dedicated for WPT \cite{HuangLauArXiv:EnablingWPTinCellularNetworks:2013} are deployed for wirelessly powering dense backscatter D2D  links  and each node transmits data by reflecting and modulating  the carrier signals sent by PBs.  This work aims at modeling and analyzing the performance of a WP-BackCom network. To be specific, the network is modeled as  Poisson cluster processes and its coverage and capacity are analyzed using stochastic geometry.

\subsection{Backscatter Communication}
Backscatter communication refers to a design where a radio device transmits data via reflecting and modulating an incident radio frequency (RF) signal by adapting the level of antenna impedance mismatch to vary the reflection coefficient and furthermore harvests energy from the signal for operating the circuit  \cite{42stockman1984communication, 1boyer2014backscatterComm.}. Backscatter devices do not require oscillators to generate carrier signals that are obtained from the air instead. Furthermore, using the simple analog modulation scheme, the device requires no analog-to-digital converters (ADCs) used in the case of digital modulation. As a result of these features, a backscatter transmitter consumes power orders-of-magnitude less than a conventional radio. Traditionally, backscatter communication is widely used  in the application of radio-frequency Identification (RFID) where a reader powers and communicates with a RFID tag over a short range no more than typically of several meters \cite{umeda2006950, chawla2007overview}.  Recent years have seen the development of more sophisticated backscatter communication systems and advancement in the underpinning communication theory \cite{G.Yang:backscattercommun.:2015, Wang:2012:EfficientReliable, saad2014physical, kimionis2014increased, dardari2010ultrawide}.  To reduce tag complexity, estimation of forward (for WPT) and backward (for information transfer) channels is usually performed at the reader that transmits and receives a training signal over different antennas. The close-loop propagation makes it difficult to estimate the two channels separately, which are required for transmit and receive beamforming. A  channel-estimation algorithm for energy  beamforming is proposed in \cite{G.Yang:backscattercommun.:2015} that alleviates this difficulty by exploiting the structure of a backscatter channel to improve WPT efficiency.  Multiple-access in a backscatter network is studied in \cite{Wang:2012:EfficientReliable}. By treating collision of backscatter nodes as a sparse code, a novel  approach  is proposed in \cite{Wang:2012:EfficientReliable} that resolves the collisions using sparse sensing algorithm and rateless-code decoding algorithm. The work in \cite{saad2014physical} addresses the issue of security in backscatter communication by proposing the technique of injecting a noise-like signal into the carrier signal sent by a reader to avoid eavesdropping of the backscatter signal.  In \cite{kimionis2014increased}, a variation of the backscatter communication system with an improved range is proposed, which has the key feature of detaching the carrier transmitter from the reader and placing it near the tag. Based on the new system design, the bit-error-rate performance of various modulation schemes is analyzed. Last, a new system based on ultra-wide bandwidth (UWB) is developed in \cite{dardari2010ultrawide}. The advantage of such a system is that based on time-hopping spread spectrum, the huge processing gain arising from signaling using ultra-sharp pulses eliminates interference between coexisting backscatter communication links.

The traditional reader-tag configuration   is unsuitable for IoT since typical nodes are energy-constrained and may not be able to wirelessly power other nodes for communications over sufficiently  long distances. This motivated the design of a backscatter communication system powered by RF energy harvesting,  where the transmission of a backscatter node relies on harvesting energy and reflecting incident RF signals from the ambient environment such as TV, Wi-Fi and cellular signals \cite{45liu2013ambient, Kellogg:2014:WIFIBackscatter, Liu2014:Enabling}. From the perspective of energy-source  systems (e.g., TV, Wi-Fi or cellular systems),  the effect of  coexisting  backscatter transmitters  is to generate addition multi-paths.  The systems have two drawbacks. First, to avoid mutual interference between coexisting systems, backscatter links have much lower data rates (using the on/off-keying modulation) than the energy-source links. Second,  backscatter communication networks based on ambient RF energy harvesting do not have scalability due to their dependance on other  networks as energy sources. Thus, they  may not be suitable for implementing large-scale dense  IoT with relatively high rates. This motivates the design of WP-BackCom network  architecture where WPT can deliver power much higher than that by energy harvesting and low-complexity backhaul-less PBs are widely deployed to power dense passive D2D links.

\subsection{Modeling the WP-BackCom Network}

To study the performance of a large-scale WP-BackCom network, we adopt stochastic geometry to design and analyze wireless networks (a survey can be found in e.g., \cite{HaenggiAndrews:StochasticGeometryRandomGraphWirelessNetworks}). Among various types of spatial point processes, \emph{Poisson cluster process} (PCP), where \emph{daughter} points form random clusters centered at points from a \emph{parent} \emph{Poisson point process} (PPP),  are commonly used for modeling wireless networks with random cluster topologies arising from geographical factors or protocols for medium access control \cite{GantiHaenggi:OutageClusteredMANET:2009, gulati2010statistics}. In particular, in recent work on random  networks, PCPs have been frequently used to model the phenomenons  of user clustering at hotspots \cite{chun2015modeling} and the clustering of small-cell base stations (BSs) around macro-cell BSs \cite{suryaprakash2015modeling} for heterogeneous networks or distributed-node clustering in D2D networks \cite{Chong:2016Modeling}. In this work, the  WP-BackCom network  is also modeled as a PCP where PBs form the parent PPP and backscatter nodes are the clustered daughter points. The clustering phenomenon arises from the fact that only nodes lying within a given distance from the nearest PB can harvest sufficient energy for operating their circuits and powering their transmissions. In other words, active backscatter nodes are effectively  attracted to PBs, thereby forming clusters with PBs at their centers. The location of a node can be uniformly distributed in a circle or normally scattered with a given variance around the centered affiliated PB, yielding the Matern and Thomas cluster processes, respectively. The use of dedicated stations similar to PBs instead of relaying on readers  for powering backscatter nodes is also observed in \cite{kimionis2014increased} to improve the link bit-error-rate performance.
Relying on WPT from PBs, nodes' transmission power depends on their distances from the associated PBs. In contrast, in the conventional network models,  transmission power of BSs/nodes is independent of their locations. The location-dependent transmission power in the WP-BackCom network as well as other practical factors (e.g., circuit power consumption, backscatter duty cycle, reflection coefficient, etc.) introduce new challenges for network performance analysis.

Recently, stochastic geometry has been also applied to model large-scale energy harvesting (including WPT) networks  building on existing network architectures including cellular networks \cite{HuangLauArXiv:EnablingWPTinCellularNetworks:2013, che2015spatial}, relay networks \cite{krikidis2014simultaneous, mekikis2014wireless}, heterogeneous networks \cite{Dhillon:2014Fundamental}, cognitive networks \cite{Hossain:2015Cognitive},  and general networks with mutually repulsive nodes \cite{Wang:2016Self}. In particular, the WP-BackCom network  is  similar to cellular networks with WPT considered in \cite{HuangLauArXiv:EnablingWPTinCellularNetworks:2013, che2015spatial} where PBs are deployed to power passive nodes'  transmissions. Nevertheless, the current work faces new theoretical challenges arising from a new network topology based on a PCP instead of PPPs in the prior work. Furthermore, practical factors arising from backscatter also introduce a new dimension for network performance optimization.

\subsection{Summary of Contributions and Organization}

The  mentioned model of random network topology corresponds to a \emph{normal}  WP-BackCom network with the nodes density exceeding that of PBs. Consider the models of backscatter transmission and channels. Time is divided into slots, each of which is further divided into mini-slots. For an arbitrary slot, each transmitting node performs backscatter transmission as well as energy harvesting in a randomly selected mini-slot and only energy harvesting in other mini-slots. The fraction of time used for transmission is referred to as the \emph{duty cycle}. Furthermore, for transmission, the fraction of backscattered incident energy is called the \emph{reflection coefficient}. The condition required for backscatter transmission is that the total energy harvested by a transmitting node within each slot exceeds the energy consumed by the circuit, introducing a \emph{circuit-power constraint}. Next, a PB beams  a carrier signal with fixed power to each target node over a free-space channel with only path loss. Moreover, the pairs of communicating nodes have unit distances and their channels have both fading and path loss.

To the best of our knowledge, this paper presents the first attempt to model and analyze a large-scale backscatter communication network using stochastic geometry.  The theoretic  contributions of this paper are summarized as follows.

\begin{enumerate}
\item The interference distribution  is  analyzed. Given the PCP, the  interference power for a typical receiver is decomposed into \emph{intra-cluster} and \emph{inter-cluster} interference.  Their characteristic functionals are derived for the normal WP-BackCom network and their product gives that for the total interference power.

\item The distribution of node-transmission power is also investigated. A key variable characterizing the transmission power of a transmitting node is the probability, called the \emph{power-outage probability},  that the circuit-power constraint is not satisfied, corresponding to insufficient energy for turning on the node. The probability is derived in closed-form.

\item Consider the network coverage measured by the metric of \emph{success probability} defined as the probability that the transmission over a typical D2D link is successful. The results in 1) and 2) are applied to derive a lower bound on the success probability. The result reveals that the success probability is a convex function of the reflection coefficient and can be thus maximized over this parameter.

\item Consider the network  transmission capacity defined as the density of reliable and active backscatter D2D links. In the regime of almost-full network coverage (with close-to-one success probability), the capacity is shown to be a convex function of the duty cycle and a monotone decreasing function of the reflection coefficient for most of range. Based on the result, the network capacity is optimized.

\item In addition, we also consider a network with \emph{dense} PBs. The corresponding model is modified from the counterpart for the normal network by denoting the parent points as transmitting nodes and the daughter points as PBs. Then the results for the normal network are extended to the network with dense PBs.

\end{enumerate}

The reminder of this paper is organized as follows. The mathematical models and performance metrics are described in Section \uppercase\expandafter{\romannumeral2}. Section \uppercase\expandafter{\romannumeral3} presents the interference characteristic functionals and transmission power distribution for WP-BackCom networks. The network coverage and capacity are analyzed in Section \uppercase\expandafter{\romannumeral4} and \uppercase\expandafter{\romannumeral5}, respectively. Simulation results are presented  in Section \uppercase\expandafter{\romannumeral6} followed by concluding remarks  in Section \uppercase\expandafter{\romannumeral7}.

{\bf Notation:} Given $X\in\mathds{R}^2$, $|X|$ denotes the Euclidean distance from $X$ to the origin. For a set $\mathcal{X}$, $|\mathcal{X}|$ yields its  cardinality. The set-subtraction operator is denoted as  $\backslash$.

\begin{table}[t]

\caption{Summary of Notation}\label{tab:para_system}
\begin{center}
\vspace{-5mm}
\begin{tabular}{|c!{\vrule width 1.5pt}l|}
\hline
Notation&Meaning\rule{0pt}{3mm}\\
\hhline{|=|=|}
$\Pi_{\text{pb}}$, $\lambda_{\text{pb}}$& The PPP modeling PBs in a normal WP-BackCom network, its density\\
\hline
$\tilde{\Phi}_{\text{nd}}$, $\lambda_{\text{pb}}\bar{c}$& The PCP modeling transmitting nodes in the normal WP-BackCom network, its density\\
\hline
$\Phi_{\text{nd}}$, $\lambda_{\text{pb}}\bar{c}D$& The PCP modeling backscatter nodes in the normal WP-BackCom network, its density\\
\hline
$\Pi_{\text{nd}}$, $\lambda_{\text{nd}}D$& The PPP modeling backscatter nodes in the WP-BackCom network with dense PBs, its density\\
\hline
$\Phi_{\text{pb}}$, $\lambda_{\text{nd}}\bar{m}$& The PCP modeling PBs in the  WP-BackCom network with dense PBs, its density\\
\hline
$Y_0, X_0, Z_0$& Typical PB, transmitting node, intended receiver\\
\hline
$\eta$, $g$& Transmission power, beamforming gain for each PB  \\
\hline
$P_X$& Received power at a backscatter node $X$ \\
\hline
$\alpha_1$, $\alpha_2$& Path-loss exponents for the WPT links, the D2D backscatter links \\
\hline
$h$& Rayleigh fading coefficient \\
\hline
$D$, $\beta$& Duty cycle, backscatter reflection coefficient \\
\hline
$P_c$& Circuit power consumption\\
\hline
$\theta$& SIR threshold for success D2D backscatter transmission\\
\hline
$P_s$, $p_0$, $C$& Success probability of backscatter communication, power-outage probability,  network transmission capacity \\
\hline
\end{tabular}
\vspace{-6mm}
\end{center}
\end{table}

\section{Mathematical Models and Metrics}

The mathematical models and performance metrics are discussed in the following sub-sections. The notation used this work is summarized in  Table~\ref{tab:para_system}.

\subsection{Spatial Network Models}

\subsubsection{Normal WP-BackCom Network}
Consider the normal WP-BackCom network where the PB density is smaller than that of backscatter nodes.  This normal PBs model is constructed using a PCP as follows. Let $\Pi_{\text{pb}} = \{Y_0, Y_1, \cdots\}$ denote a PPP in the horizontal plane with density $\lambda_{\text{pb}}$ modeling the locations of PBs. Consider a cluster of mobile transmitting nodes centered at the origin, denoted as $\widetilde{\mathcal{N}}= \{X_0, X_1, \cdots, X_N\}$. The number of nodes, $N$, is a Poisson random variable (r.v.) with mean $\bar{c}$. The r.v.,   $X_n\in \mathds{R}^2$,  represents the location of the corresponding  node  and $\{X_n\}$ are independent and identically distributed (i.i.d.). For an arbitrary r.v. $X_n$, the direction is isotropic and the distance to the origin, $|X_n|$,  has one of two possible probability density functions (PDFs), resulting in the Matern and Thomas cluster processes. Let the function $f: \mathds{R}^{+}\rightarrow \mathds{R}^{+}$ denote the PDF of the distance from a cluster member's location  to the cluster center, denoted as $r$,  which is defined as  follows:
\begin{align}
\text{(Matern c.p.)}   \quad f(r) &=
   \begin{cases}
       \frac{1}{\pi a^{2}}, &0 \leq r \leq a, \\
     0, &\mbox{otherwise,}
   \end{cases}\label{Eq:PDF:Matern}\\
\text{(Thomas c.p.)}\quad     f(r)&=\frac{1}{2\pi \sigma^{2}}\exp\left ( -\frac{r^{2}}{2\sigma^{2}} \right), \label{Eq:PDF:Thomas}
\end{align}
where $a$ and $\sigma^2$ are positive constants representing the cluster radius and the variance, respectively.   Let $\{\widetilde{\mathcal{N}}_Y\}$ denote a sequence of clusters constructed by generating an  i.i.d. sequence of clusters having the same distribution as $\widetilde{\mathcal{N}}$ and translating them to be centered at the points $\{Y\}\in \Pi_{\text{pb}}$.
Then the process of transmitting nodes, denoted as $\tilde{\Phi}_{\text{nd}}$, can be written as $\tilde{\Phi}_{\text{nd}} = \bigcup_{Y\in \Pi_{\text{pb}}} \widetilde{\mathcal{N}}_Y$. The density of $\tilde{\Phi}_{\text{nd}}$ is $\lambda_{\text{pb}} \bar{c}$. The link from a PB to an intended transmitting node is called a \emph{WPT link}. Each transmitting node is paired with an intended receiving node that is located at a unit distance and in an isotropic direction, forming a \emph{D2D backscatter transmission link}. Fig.~\ref{generalPB} shows two network realizations generated based on the Matern and Thomas cluster processes, which helps to visualize the topology of Poisson clustered network as well as the WPT links.


\begin{figure}[t]
\centering
\subfigure[Normal PBs model: Matern cluster process]{\includegraphics[width=8cm]{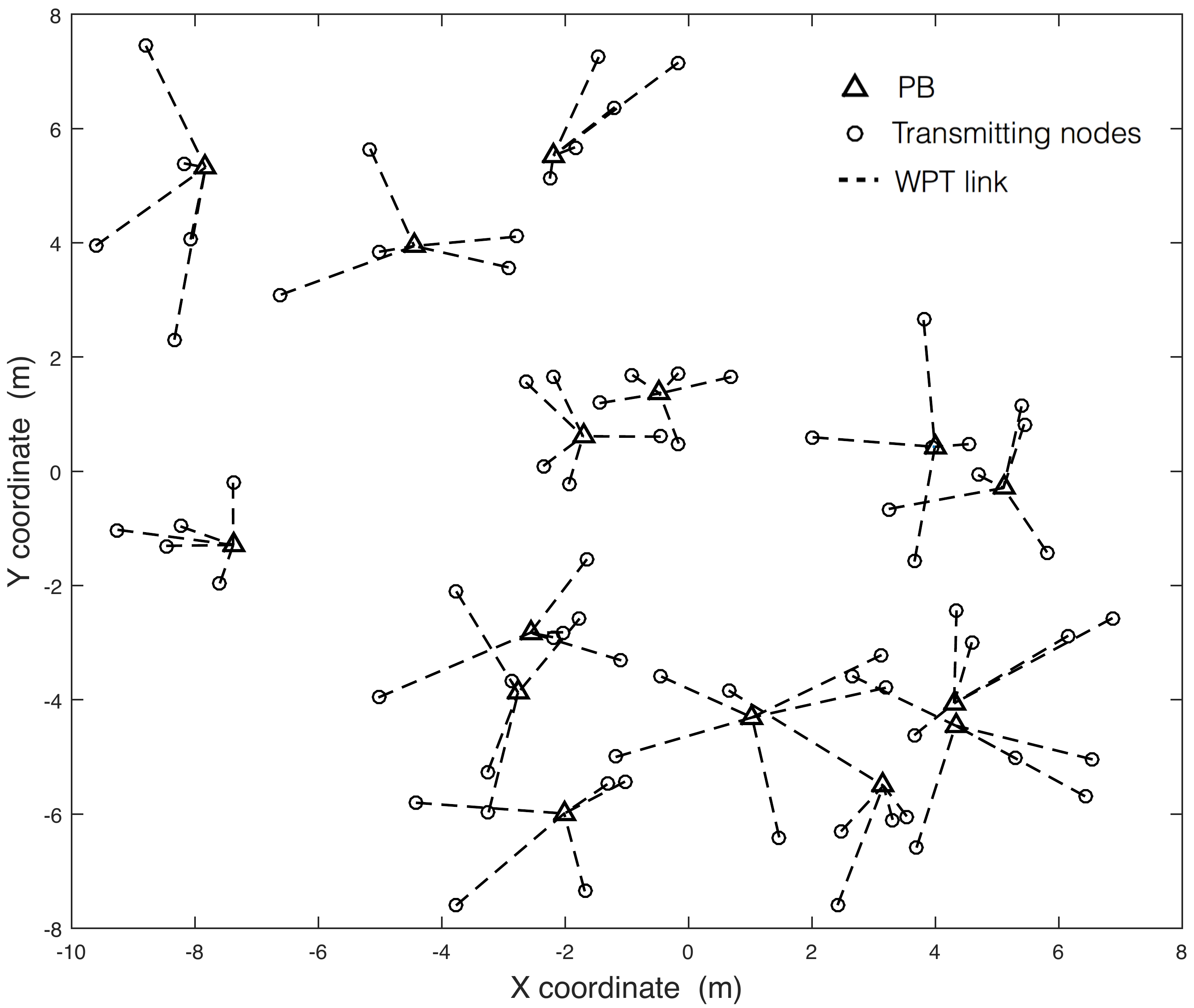}}
\subfigure[Normal PBs model: Thomas cluster process]{\includegraphics[width=8cm]{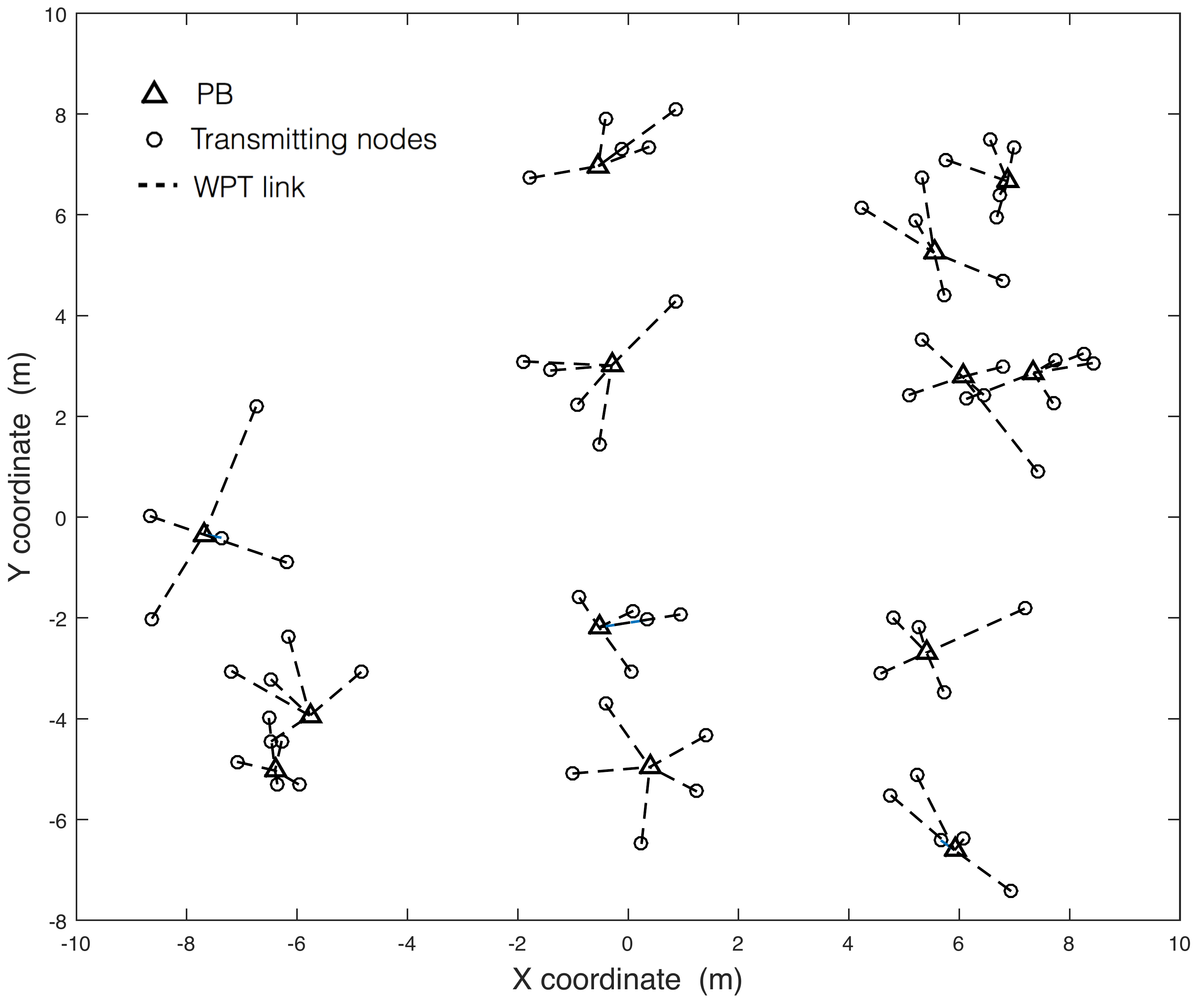}}
\caption{The spatial distribution of the normal WP-BackCom network modeled using the (a) Matern cluster process and (b) Thomas cluster process.}\label{generalPB}
\end{figure}

Time is divided into slots with unit duration.  Each slot is further divided into $M$ mini-slots. A transmitting node randomly selects a single mini-slot to transmit signal by  backscattering and the selection remains unchanged over slots. The selections by  different nodes are assumed independent.  This divides each slot into a backscatter phase and a waiting  phase of durations $1/M$ and $(1 - 1/M)$, respectively. The duty cycle,  denoted as $D$, is given as $D = 1/M$. A transmitting node in a backscatter phase is called a \emph{backscatter node}. Then the backscatter-node process, denoted as $\Phi$, and a cluster of  backscatter nodes centered at $Y$, denoted as $\mathcal{N}_Y$, can be obtained from $\tilde{\Phi}_{\text{nd}}$ and $\widetilde{\mathcal{N}}_Y$ by independent thinning. As a result, $\Phi_{\text{nd}}$ has the density of $\lambda_{\text{pb}} \bar{c}D$ and the expected number of nodes in $\mathcal{N}_Y$ is $\bar{c} D$.

\subsubsection{WP-BackCom Network with Dense PBs}
Consider the case where the PB density exceeds that of backscatter nodes. The corresponding network model with dense PBs is modified from the normal network model by switching the locations of the PBs and transmitting nodes. Specifically, the model comprises a PCP where the parent PPP, denoted as $\Pi_{\text{nd}}$, models the random locations of backscatter nodes and the PBs are the daughter points, denoted as $\Phi_{\text{pb}}$. The density of $\Pi_{\text{nd}}$ is $\lambda_{\text{nd}} D$ with $\lambda_{\text{nd}}$ representing the density of transmitting nodes. By abusing the notation, let $\bar{m}$ represent the expected number of PBs in each cluster of $\Phi_{\text{pb}}$ and $\mathcal{N}_X$ represent a cluster of PBs centered at the backscatter node $X$. It follows that $\Phi_{\text{pb}} = \bigcup_{X\in \Pi_{\text{nd}}} \mathcal{N}_X$. Fig.~\ref{densePBsmodel} shows the dense PBs case of Thomas cluster process.

\begin{figure}[t]
\centering
\includegraphics[width=8.5cm]{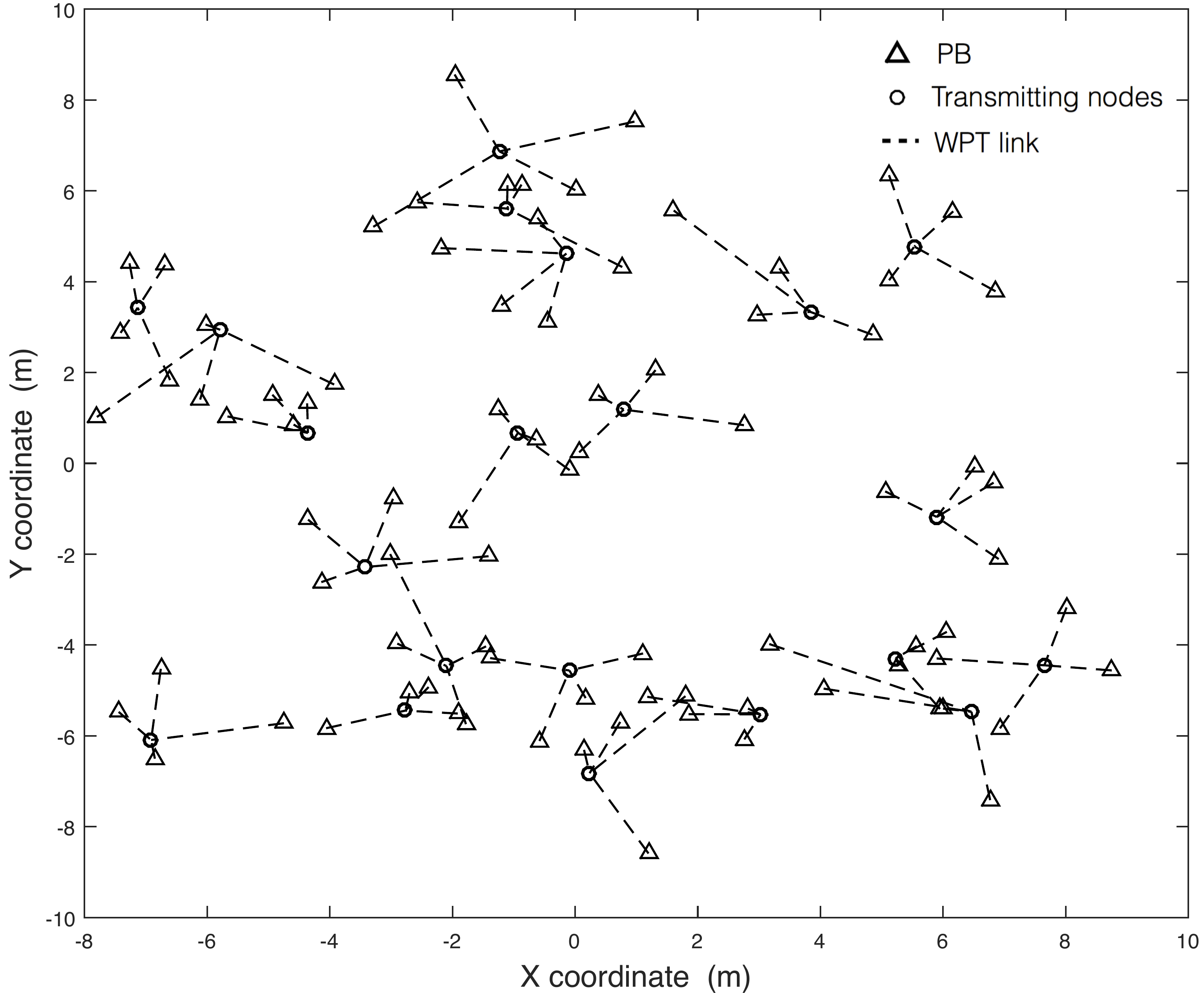}
\caption{The WP-BackCom network with dense PBs: Thomas cluster process. }\label{densePBsmodel}
\end{figure}

\subsection{Channel Model}

\subsubsection{Channel Model for the Normal WP-BackCom Network}
PBs are equipped with antenna arrays and nodes have single isotropic antennas. Each PB beams the continuous wave (CW) to nodes in its corresponding cluster. A PB serves multiple backscatter nodes  simultaneously where the total transmission power is adapted to the number of nodes to guarantee the quality-of-service. Specifically, a PB beams transmission power of $\eta$ to  each intended  node and thus the expected total transmission  power is equal to $\eta$ times the number of nodes in the corresponding cluster (the expectation is   $\bar{c}\eta$). Given beamforming and relatively short distance for efficient WPT, each WPT link is suitably modeled as a channel with path loss but no fading \cite{Huang:CuttingLastWiress:2014, HuangLauArXiv:EnablingWPTinCellularNetworks:2013}.  As a result, with  a typical PB at $Y_0$, the receive power at a typical  node $X_0$ is given as $P_{X_0} = \eta g |X_0 - Y_0|^{-\alpha_1}$ where $g>0$ denotes the beamforming gain and  $\alpha_1$ represents the path loss exponent for WPT links. Due to beamforming, it is assumed that each node harvests negligible energy from other PBs and data signals compared with that from the serving PB.  When transmitting, a  node  backscatters a fraction, called a \emph{reflection coefficient} and denoted as $\beta \in [0, 1]$, of $P_X$ such that the signal power received  at the typical receiver at $Z_0$ is $\beta P_{X_0} h_{X_0} $  \footnote{ Note that the path loss is absent due to the assumption of unit distance for each D2D link. Relaxing this assumption has no effect on the main results except for modifying the SIR threshold $\theta$ by multiplication with a constant $l^{\alpha_2}$ where $l > 0$ denotes the propagation distance for the D2D links.} where $h_{X_0}\sim\exp(1)$ models Rayleigh fading.  A transmitting node may not be able to transmit if there is insufficient energy for operating its circuit as discussed in the sequel. This constraint is represented by a function $\ell(\cdot)$, which is defined in the sequel (the following sub-section), such that the transmission power of the backscatter node $X$ can be written as $\beta \ell(P_X)$. The interference power measured at $Z_0$ can be written as
\begin{equation}
I = \sum_{X \in \Phi_{\text{nd}}\backslash \{X_0\}} \beta \ell(P_X) h_X |X - Z_0|^{-\alpha_2},\label{Eq:IntPwr:a}
\end{equation}
where $\{h_X\}$ are i.i.d. $\exp(1)$ r.v.s modeling Rayleigh fading, $\alpha_2$ represents the path loss exponent for interference (D2D) links. Recall that   $\tilde{\Phi}_\text{nd}$ and $\Phi_\text{nd}$ with $\Phi_\text{nd} \subseteq \tilde{\Phi}_\text{nd}$ represent (both active and inactive) transmitting nodes  and (active) backscatter  nodes provisioned with sufficient energy, respectively. Thus, the interference power  in \eqref{Eq:IntPwr:a} is a sum over $\Phi_\text{nd}$ instead of $\tilde{\Phi}_\text{nd}$. It is worth mentioning that the signal transmitted from PB is unmodulated carrier that appears as a DC level after down-conversion, which can be easily removed and thus not considered as interference as those modulated ones.

\begin{figure}[t]
\centering
\includegraphics[width=12cm]{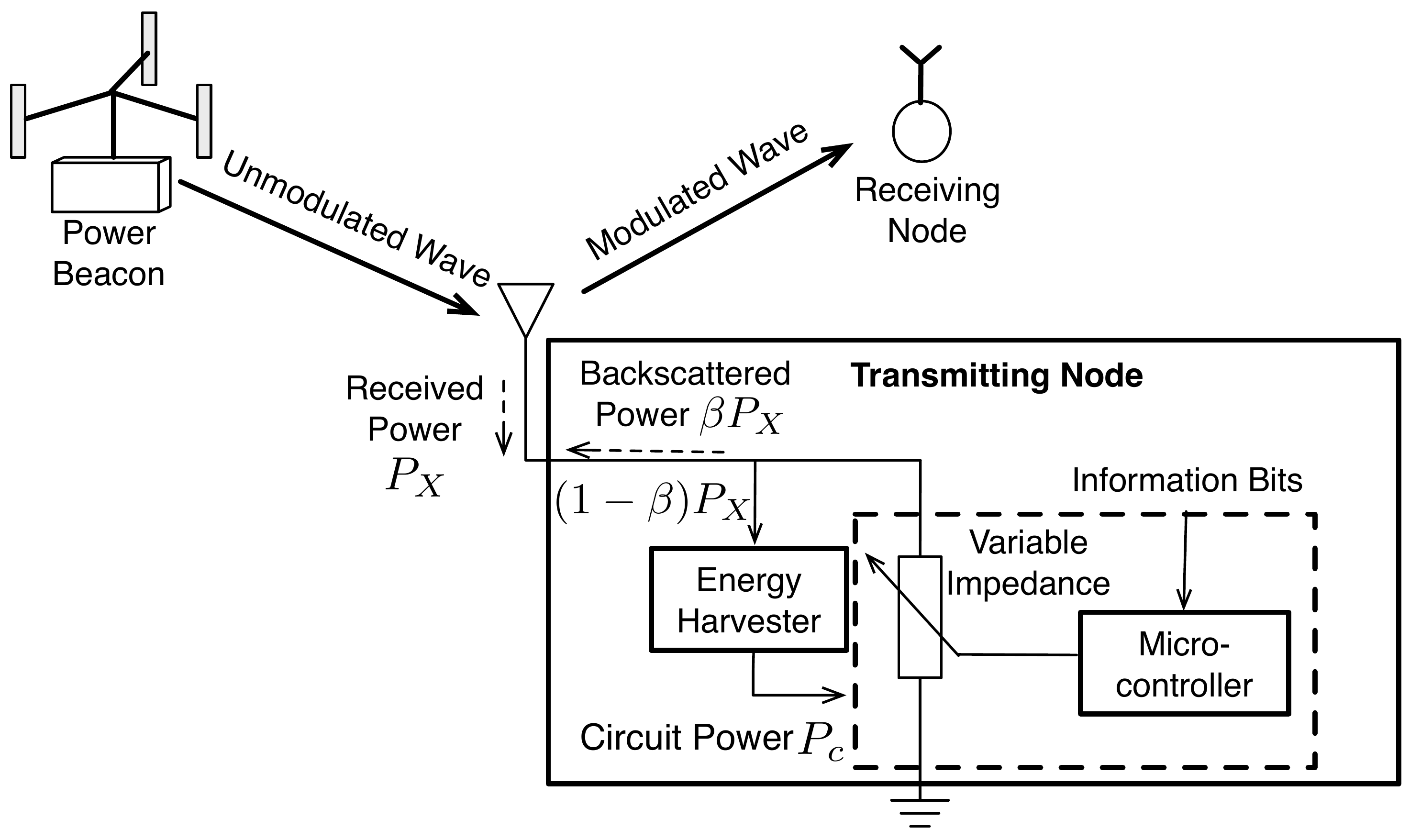}
\caption{Wirelessly powered backscatter communication.}\label{sysmodel}
\end{figure}

\subsubsection{Channel Model for the Network with  Dense PBs} The total power  received at the typical node $X_0$ is contributed by beamforming PBs in the corresponding cluster. As a result,  $P_{X_0}$ can be written as $P_{X_0} = \eta g \sum_{Y\in\mathcal{N}_{X_0}} |X_0 - Y|^{-\alpha_1}$. Moreover, the power of interference as measured at the typical receiver $Z_0$ is given as
\begin{equation}
I = \sum_{X \in \Pi_{\text{nd}}\backslash \{X_0\}} \beta \ell(P_X) h_X |X - Z_0|^{-\alpha_2}. \label{Eq:IntPwr:Dense}
\end{equation}

\subsection{Backscatter Communication Model}

The operation of WP-BackCom network is illustrated in Fig.~\ref{sysmodel}. Consider the backscatter phase of an arbitrary slot. A transmitting node adapts the variable impedance (or equivalently its level of mismatch with the antenna impedance)  shown in the figure so as to modulate the backscattered CW with information bits \cite{1boyer2014backscatterComm.}. Given a transmitting node at $X$ and the reflection coefficient $\beta$, the  backscattered power is $\beta P_X$ with the remainder $(1-\beta) P_X$ consumed by the circuit or harvested \cite{karthaus2003:fully}. Next, for the waiting  phase,  the transmitting node withholds transmission and  performs only energy harvesting. It is assumed that the circuit of each transmitting node consumes fixed power denoted as $P_c$.  To be able to transmit, a node has to harvest sufficient energy for powering the circuit, resulting in the following \emph{circuit-power constraint}: $(1-\beta) P_X D + P_X (1 - D) \geq P_c$. This gives
\begin{equation}\label{Eq:Circuit}
\text{(Circuit-power constraint)}\quad P_X \geq \frac{P_c}{1- \beta D}.
\end{equation}
Consequently, a node transmits or is silent depending on if the constraint is satisfied. Under the circuit-power constraint given in \eqref{Eq:Circuit}, a transmission power of the typical transmitting node can be written as  $\beta \ell(P_{X_0}) $ where the function $\ell(P)$ gives $P$ if $P \geq P_c/(1- \beta D)$ or else is equal to $0$. This function is also used in the interference power in \eqref{Eq:IntPwr:a} and \eqref{Eq:IntPwr:Dense}.

\subsection{Performance Metrics}

The network performance is measured by two metrics, namely the \emph{success probability} and the \emph{network transmission capacity}.  The success probability is  denoted as $P_s$. Assuming an interference limited network, the condition for successful transmission is that the receive signal-to-interference ratio (SIR) exceeds a fixed positive threshold $\theta$. The assumption of negligible noise (or equivalently interference limited network)  can be  justified as follows. The installation of PBs  for powering short-distance D2D links can result in relatively large signal power and increase the density of coexisting active D2D links that have strong mutual interference. As a rule of thumb, the network is interference limited if the expected interference power is much larger than the noise variance, resulting in conditions on network parameters. For instance, it is straightforward to derive a threshold on the PB density for the interference-limited regime.

Then the success probability is given as
\begin{equation}\label{Eq:Ps}
P_s = \Pr\l(\beta \ell(P_{X_0})h_{X_0}   \geq \theta I\r).
\end{equation}
To facilitate the analysis, the total interference will be partitioned into two parts, intra-cluster and inter-cluster interference, for characterizing its distribution in the next section.

 The other metric is network (transmission) capacity denoted as $C$ and defined as:
\begin{equation}\label{Eq:TXCap}
C = \l\{
\begin{aligned}
&B  \lambda_{\text{pb}} \bar{c} D P_s, && \text{normal network},\\
& B  \lambda_{\text{nd}} D P_s, && \text{network with dense PBs},
\end{aligned}\r.
\end{equation}
where the factor  $B$ represents the data rate per D2D link and the other factor is the spatial density of  active links with successful transmission \cite{GantiHaenggi:OutageClusteredMANET:2009, weber2010overview}. It follows that $C$ quantifies the network spatial throughput and has the unit of b/s/Hz/unit-area. Without loss of generality, $B$ is set as $1$ (b/s/Hz) to simplify notation.

\section{Interference and Transmission Power}
In this section, for the WP-BackCom network, the distributions of interference at the typical receiver and the transmission power for the typical backscatter node are analyzed. The results are used subsequently for characterizing network coverage and capacity.

\subsection{Interference Characteristic Functionals}

\subsubsection{Normal WP-BackCom Network}
Let $\mathcal{C}(s)$ with $s > 0$ denote the characteristic functional of the interference power $I$ given in \eqref{Eq:IntPwr:a}:  $\mathcal{C}(s) = \E\l[e^{-sI}\r]$.  In this section, the characteristic functional of normal WP-BackCom network model is derived. Without loss of generality, consider a typical backscatter  node at $X_0$ at the origin and the typical receiving node $Z_0 = z$. To facilitate derivation, the total interference $I$ is split in the sequel  into two components, the \emph{intra-cluster} and \emph{inter-cluster} interference, denoted as $I_a$ and $I_b$, respectively. The intra-cluster interference is caused by the interfering backscatter nodes inside the representative cluster (i.e., the cluster with typical backscatter node and receiver), and the inter-cluster interference is caused by simultaneously backscatter nodes outside the representative cluster. The separation of intra-cluster and inter-cluster  is due to their difference in distribution. Specifically, the former depends on a cluster of random points with the typical one removed and the latter on a cluster point process with the typical cluster removed. They are defined and analyzed separately to yield insight into the structure of interference distribution.

Mathematically, $I = I_a + I_b$  where
\begin{align}
I_a &= \sum_{X\in \mathcal{N}_{0} \backslash \{ X_0\}} \beta \ell(P_X)  h_X |X - z |^{-\alpha_2}, \label{Eq:IntraInt}\\
I_b &= \sum_{Y\in \Pi_{\text{pb}}\backslash \{Y_0\}  }\sum_{X\in \mathcal{N}_{Y} } \beta \ell(P_X)  h_X |X-z|^{-\alpha_2}. \label{Eq:InterInt}
\end{align}
Note that in $\eqref{Eq:InterInt}$, the first summation is over all other PBs not affiliated with the typical backscatter node (corresponding to clusters of interferers) and the second summation is over the cluster of interferers  centered at the PB $Y$.

The characteristic functionals of $I_a$ and $I_b$ are denoted as $\mathcal{C}_a(s)$ and $\mathcal{C}_b(s)$, respectively, which are defined similarly as $\mathcal{C}(s)$.  They are derived as shown in the following two lemmas.

\begin{lemma}[Intra-cluster Interference in the Normal Network]\label{Lem:IntraInt}\emph{Given $s \geq 0$, the characteristic functional of the intra-cluster interference power $I_a$ is given as
\begin{equation}
\mathcal{C}_a(s) = \int_{\mathds{R}^2}\exp\Big(-\bar{c} D q(s, y, z) \Big)f(|y|) d y ,
\end{equation}
where
\begin{align}
q(s, y, z) = \int_{\mathds{O}}\frac{1}{1 + (s \beta \eta g)^{-1}|x|^{\alpha_{1}}|x + y -z|^{\alpha_{2}}}f(|x|)d x,
\end{align}
$f(\cdot)$ is specified in (\ref{Eq:PDF:Matern}) and (\ref{Eq:PDF:Thomas}), and the set $\mathds{O}$ arising from the \emph{circuit-power constraint} is defined as
\begin{equation}
\mathds{O} = \l\{x \in \mathds{R}^2\mid |x| \leq \l(\frac{\eta g(1-\beta D)}{P_c}\r)^{\frac{1}{\alpha_1}}\r\}. \label{eq:effectiveI}
\end{equation}
}
\end{lemma}
\begin{proof}
See Appendix \ref{appendix:lemma:1}.
\end{proof}

\begin{lemma}[Inter-cluster Interference in the Normal Network]\label{Lem:InterInt}\emph{Given $s \geq 0$, the characteristic functional of the inter-cluster interference power $I_b$ is given as
\begin{equation}
\mathcal{C}_b(s) =\exp\Bigg(-\lambda_{\text{pb}} \int_{\mathds{R}^2} \Big(1 - e^{-\bar{c} D q(s, y, z)}  \Big)dy\Bigg),
\end{equation}
where $q(s, y, z)$  and the set $\mathds{O}$ are defined in Lemma~\ref{Lem:IntraInt}.
}
\end{lemma}

\begin{proof}
See Appendix \ref{appendix:lemma:2}.
\end{proof}

\subsubsection{Network with Dense PBs} In this model, the interferes are distributed as a homogeneous PPP instead of a PCP as in the normal network. Thus, it is unnecessary to decompose the total interference into intra- and inter-cluster interference. The total interference power can be written as
\begin{align}
I^{'} &= \sum_{X \in \Pi_{\text{nd}}\backslash \{X_{0}\}} \beta \ell (P_{X}) h_{X} |X-z|^{-\alpha_{2}} \nonumber \\
&= \sum_{X \in \Pi_{\text{nd}}\backslash \{X_{0}\}} \beta\ell \l( \sum_{Y \in \mathcal{N}_X} \eta g |Y-X|^{-\alpha_{1}} \r)  h_{X} |X-z|^{-\alpha_{2}}.
\end{align}
It can be observed that the received power at each node is a sum over a cluster of PBs instead of a single PB in the normal network. For analytical tractability, relaxing the circuit-power constraint allows nodes that are previously inactive under this constraint to transmit, resulting in the following upper bound on the interference power:
\begin{equation}
I^{'}\leq  \beta \eta g \sum_{X \in \Pi_{\text{nd}}\backslash \{X_{0}\}} \l( \sum_{Y \in \mathcal{N}_X}  |Y-X|^{-\alpha_{1}} \r)  h_{X} |X-z|^{-\alpha_{2}}. \label{De:InterDense}
\end{equation}
Notice that since each node is charged by multiple PBs, it is likely that the circuit-power constraints are satisfied at the majority of nodes and thus the lower bound in \eqref{De:InterDense} is tight. Using this inequality, a lower bound on the  characteristic functional of $I^{'}$, denoted as $\mathcal{C}^{'}(s)$, is obtained as shown in the  following lemma.

\begin{lemma}[Interference in the Network with Dense PBs] \label{Lem:densePBmodel} \emph{ Given $s\geq0$, the characteristic functional of the interference power $I^{'}$ satisfies
\begin{equation}
\mathcal{C}^{'}(s)  \geq  \exp\Bigg(- \lambda_{\text{nd}} D \int_{\mathds{R}^2} \Big( 1- e^{-\bar{m} u(s, x, z)} \Big) dx\Bigg),
\end{equation}
where
\begin{align}
u(s,x,z) = 2\pi \int^{\infty}_{0} \frac{1}{1+(s\beta \eta g)^{-1} r^{\alpha_{1}}|x-z|^{\alpha_{2}}} f(r) r dr
\end{align}
and $f(\cdot)$ is specified in (\ref{Eq:PDF:Matern}) and (\ref{Eq:PDF:Thomas}).
}
\end{lemma}

\begin{proof}
See Appendix \ref{appendix:lemma:5}.
\end{proof}

\begin{remark}\emph{Comparing Lemma~\ref{Lem:InterInt} and \ref{Lem:densePBmodel}, the lower bounds on $\mathcal{C}_b(s)$ and $\mathcal{C}^{'}(s)$ have similar expressions. The key difference is that the PB density ($\lambda_{\text{pb}}$) and expected number of backscatter nodes per cluster ($\bar{c}D$) in  the former are replaced by the backscatter-node density ($\lambda_{\text{nd}} D$) and expected number of PBs per cluster ($\bar{m}$) in the latter. The similarity arises from  that both network models are based on  PCPs where the locations of backscatter nodes and PBs are swapped.
}
\end{remark}

\subsection{Transmission-Power Distribution}

\subsubsection{Normal WP-BackCom Network}

Under the circuit-power constraint, there exists a threshold on the separation distance between a pair of PB and affiliated backscatter node:
\begin{equation}
d_0 = \l[\frac{\eta g (1 - \beta D)}{P_c}\r]^{\frac{1}{\alpha_1}}, \label{Eq:Th:Dist}
\end{equation}
such that  the node's transmission power is zero if the distance exceeds the threshold. Then transmission power (i.e., backscattered power) of the typical backscatter node, denoted as $P_t$, is given as $P_t =  \beta \eta g |X_0 - Y_0|^{-\alpha_1}$ if $|X_0 - Y_0| \leq d_0$ or otherwise $P_t = 0$. The event of $P_t = 0$ corresponds that of circuit power outage.
It follows that the \emph{power-outage probability}, denoted as $p_0$, can be written as
\begin{equation}
p_0 = \Pr(P_t = 0) = 2\pi \int_{d_0}^\infty f(r) r dr.   \label{Eq:SigDist:1}
\end{equation}
For the case where the circuit-power constraint is satisfied,
\begin{equation}
\Pr(P_t \geq  \tau) = 2\pi \int\limits_0^{(\beta \eta g/\tau)^{\frac{1}{\alpha_1}}} f(r) rdr, \ \  \tau \geq  \frac{\beta P_c}{1 - \beta D}. \label{Eq:SigDist:2}
\end{equation}
Substituting the PDFs in \eqref{Eq:PDF:Matern} and \eqref{Eq:PDF:Thomas} into \eqref{Eq:SigDist:1} and \eqref{Eq:SigDist:2} gives the following results.

\begin{lemma}[Node Transmission Power for the Normal Network] \label{Lem:SigDist}\emph{The transmission power of a typical backscatter node has support of $\{0\}\cup [\frac{\beta P_c}{1-\beta D}, \infty]$.  The power-outage probability, $p_0$, and the complementary cumulative distribution function (CCDF), denoted as $\bar{F}_t$,  are given as follows.
\begin{itemize}
\item[--] (Matern cluster process)
\begin{align}
p_0  &= \l\{
\begin{aligned}
& 1 -\l(\frac{d_0}{a}\r)^2, && d_0 <  a, \\
& 0, && \text{otherwise.}
\end{aligned}\r.\nn\\
\bar{F}_t(\tau) = \Pr(P_t \geq \tau ) &= \l\{
\begin{aligned}
& \frac{1}{a^2}\l(\frac{\beta \eta g}{\tau}\r)^{\frac{2}{\alpha_1}}, && \tau >   \frac{\beta \eta g }{a ^{\alpha_1}}\\
& 1, && \text{otherwise,}
\end{aligned}\r.\nn
\end{align}
with  $\tau \in [\frac{\beta P_c}{1-\beta D}, \infty]$.
\item[--] (Thomas cluster process)
\begin{align}
p_0 &= \exp\l(-\frac{d_0^2}{2\sigma^2}\r), \nn\\
\bar{F}_t(\tau) &= 1 - \exp\l( - \frac{1}{2\sigma^2}\l(\frac{\beta\eta g}{\tau}\r)^{\frac{2}{\alpha_1}}\r), \nn
\end{align}
with  $\tau \in [\frac{\beta P_c}{1-\beta D}, \infty]$.
\end{itemize}
}
\end{lemma}
 A sanity check is as follows. The distance threshold $d_0$ in \eqref{Eq:Th:Dist} is a monotone decreasing  function of both $\beta D$ and $P_c$. The reason is that increasing the duty cycle and reflection coefficient leads to less harvested energy thus adding the probability of circuit-power outage and improving the circuit power conumption has the same effect. Consequently, the power-outage probability decreases with increasing $d_0$ for both cases in Lemma~\ref{Lem:SigDist}. Next, the CCDFs in Lemma~\ref{Lem:SigDist} are observed to be independent of $D$ but increase with growing $\beta$. The reason is that conditioned on the node transmitting, the transmission power depends only on the incident power from the PB scaled by  $\beta$ but is independent of the duty cycle.

\begin{remark}[Effects of $P_c$ on Power Outage]\emph{Let $\{c_n\}$ denote a set of constants. For the  Matern cluster process, the power-outage probability $p_0$ is a monotone increasing function of  the circuit power $P_c$ and scales as $\l(1 - c_1 P_c^{-\frac{2}{\alpha_1}}\r)$. One can see that increasing the path loss exponent $\alpha_1$ alleviates the negative effect of increasing $P_c$ on power outage. In contrast, the scaling law for the Thomas cluster process is $\exp\l(- c_2 P_c^{-\frac{2}{\alpha_1}}\r)$. Comparing the two scaling laws reveals that the effect of increasing $P_c$ is less severe in the latter model where nodes tend to be be nearer to PBs and thus the WPT loss is smaller.
}
\end{remark}

\begin{remark}[Effects of $\beta$ and $D$ on Power Outage]\emph{For the  Matern cluster process, $p_0$ also grows  with the increasing product of the reflection coefficient $\beta$ and duty cycle $D$. For $\beta D \approx 0$, the scaling law  is $\frac{c_3\beta D}{\alpha_1} + c_4$. This suggests a tradeoff between $\beta$ and $D$ and the positive effect of having increasing the  path loss exponent.  The scaling law for the Thomas cluster process is $c_6\exp\l(\frac{c_5 \beta D}{\alpha_1}\r) \approx \frac{c_7\beta D}{\alpha_1} + c_6$ that is similar to the counterpart for the Matern cluster process.
}
\end{remark}

\subsubsection{Network with Dense PBs}

Given that the circuit-power constraint is satisfied, the transmission power, denoted as  $P_t'$,  at the typical backscatter node $X_0$, is a fraction of the incident power $P_{X_0}'$ that follows the compound Poisson distribution: $P_{X_0}' = \eta g \sum_{n}^N d_n^{-\alpha_1}$ where $N$ denotes  the number of daughter points in the cluster $\mathcal{N}_{X_0}$ and the set of i.i.d. r.v.s $\{d_n\}$ re-denotes $\{|X_0 - Y|\mid Y \in \mathcal{N}_{X_0}\}$ to simplify notation. The distribution of $\{d_n\}$ is specified by $f(r)$ in \eqref{Eq:PDF:Matern} and \eqref{Eq:PDF:Thomas}. The distribution function of $P_{X_0}'$ (or equivalently $P_t'$) has no closed-form. The following analysis  focuses on characterizing the power-outage probability defined as
\begin{equation}
p_0' = \Pr(P_t' = 0) = \Pr\l(P_{X_0}' \leq \frac{P_c}{1-\beta D}\r).
\end{equation}
Applying Chernoff bound gives
\begin{equation}\label{Eq:Chernoff}
p_0'  \leq \min_{\mu > 0} \E\l[e^{-\mu \sum_{n}^N d_n^{-\alpha_1} }\r]e^{\frac{\mu P_c}{(1-\beta D)\eta g}}.
\end{equation}
By deriving the characteristic functional of $\sum_{n}^N d_n^{-\alpha_1} $,  $p_0'$ can be upper bounded as shown below.

 \begin{lemma} [Node Transmission Power for the Network with Dense PBs] \label{Lem:SignalDense} \emph{ The power-outage probability $p_{0}^{'}$  satisfies the following.
  \begin{itemize}
  \item[--] (Matern cluster process)
  \begin{align}
  p_{0}^{'} \leq \exp\Bigg(\frac{\mu^* P_c}{(1 - \beta D)\eta g} -  \frac{\bar{m}}{a^{2}}\int^{a^{2}}_{0}\l(1- e^{-\mu^* t^{-\frac{\alpha_{1}}{2}}} \r) dt \Bigg), \label{Iq:maternChernoff}
  \end{align}
  where   $\mu^* > 0 $ solves
  \begin{align}
 \int^{a^{2}}_{0} t^{-\frac{\alpha_{1}}{2}} e^{-\mu^* t^{-\frac{\alpha_{1}}{2}}} dt = \frac{a^2 P_c (\eta g)^{-1}}{\bar{m}(1 - \beta D)}. \label{Eq:optimalMuMatern}
  \end{align}
  \item[--] (Thomas cluster process)
  \begin{align}
  p_{0}^{'} \leq \exp \Bigg(\frac{\mu^* P_c}{(1 - \beta D)\eta g} - \bar{m} \int^{\infty}_{0} e^{-t} \l(1 - e^{-t -\mu^* (\sqrt{2}\sigma)^{-\alpha_{1}}t^{-\alpha_{1}/2}}\r)  dt\Bigg), \label{Iq:thomasChernoff}
  \end{align}
  where $\mu^* > 0 $ solves
  \begin{align}
\int^{\infty}_{0} \frac{e^{-t -\mu^* (\sqrt{2}\sigma)^{-\alpha_{1}}t^{-\alpha_{1}/2}}}{(\sqrt{2}\sigma)^{\alpha_{1}} t^{\frac{\alpha_{1}}{2}} } dt  = \frac{P_c(\eta g)^{-1}}{\bar{m}(1 - \beta D)}. \label{Eq:optimalMuThomas}
  \end{align}
  \end{itemize}}
  \end{lemma}

\begin{proof}
See Appendix \ref{appendix:proposition:2}.
\end{proof}

\begin{remark}[Effects of Backscatter Parameters on Power Outage]\emph{By approximating the power-outage probability $p_0'$ using its upper bounds in Lemma~\ref{Lem:SignalDense}, $p_0'$ scales as $c\exp\l(\frac{\mu^* P_c}{(1 - \beta D)\eta g}\r)$ with $c$ being a constant for both the Matern and Thomas cluster processes.
}
\end{remark}

\begin{remark}[Effects of PB density on Power Outage]\emph{For both the Matern and Thomas cluster processes, the power-outage probability is observed to diminish at least exponentially with the increasing expected number of PBs serving each node, $\bar{m}$.
}
\end{remark}

\section{Network Coverage}
In this section, the coverage of the WP-BackCom network is characterized using the results derived in the preceding section.

First, consider the normal WP-BackCom network model.  The network coverage is quantified by deriving the success probability, $P_s$ defined in \eqref{Eq:Ps}, as follows. The event of successful transmission by the typical backscatter node  occurs under two conditions: 1) the circuit-power constraint in \eqref{Eq:Circuit} is satisfied and 2) under this condition,  the receive SIR exceeds the threshold $\theta$. Therefore, $P_s$ can be written as
\begin{equation}\label{Eq:Ps:Decomp}
P_s =  \Pr\l(P_t h_{X_0}   \geq \theta I\mid P_t \neq 0 \r)\Pr(P_t \neq 0).
\end{equation}
Replacing the transmission power with its minimum value gives a tight lower bound (will be discussed in simulation part) on $P_s$ as follows:
\begin{align}
P_s &\geq \Pr\l(\frac{\beta P_c h_{X_0}}{1-\beta D}  \geq \theta I \r) \Pr(P_t \neq 0) \nn\\
&= \E\l[\exp\l(-\frac{\theta I (1 -\beta D)}{\beta P_c}\r)\r](1 - p_0).  \nn
\end{align}
Then the main result of this section follows by substituting the results derived in the preceding section into the expression above.

\begin{theorem}[Coverage for the Normal Network]\label{Theo:Coverage:Typical}\emph{The \emph{success probability}, denoted as $P_s$,  is lower bounded as
\begin{equation}
P_s \geq (1-p_0) \mathcal{C}\l(\frac{\theta (1 -\beta D)}{\beta P_c}\r),
\end{equation}
where $\mathcal{C}\l(s\r)= \mathcal{C}_a\l(s\r) \mathcal{C}_b\l(s\r)$ is the interference characteristic functional with $\mathcal{C}_a$ and $\mathcal{C}_b$ given in Lemma~\ref{Lem:IntraInt} and Lemma~\ref{Lem:InterInt}, respectively, and $p_0$ is the power outage probability specified in Lemma~\ref{Lem:SigDist}.
}
\end{theorem}

\begin{remark}[Effects of $p_0$ on Network Coverage] \emph{The success probability  is observed to increase \emph{linearly} with the \emph{transmission  probability} of a backscatter node, $(1 - p_0)$, which agrees with intuition.}
\end{remark}

\begin{remark}[Effects of $D$ and $\beta$ on Network Coverage] \label{Re:Coverage}\emph{ The success probability $P_s$ can be maximized over the reflection coefficient $\beta$. A too large or a too small value for $\beta$ has a negative effect on network coverage (or the success probability). On one hand,  increasing $\beta$ not only adds the probability of circuit power outage but also scales up transmission power for each node, which can lead to strong interference. On the other hand, $\beta$ being too small leads to  weak receive signal  as well as larger effective interference set $\mathds{O}$ given in (\ref{eq:effectiveI}). Both decrease $P_s$. Next,  increasing the duty cycle $D$ causes growth of the   circuit-power outage probability and dense interferers, thereby reducing $P_s$.}
\end{remark}

Next, consider the network model with dense PBs.  Following the same procedure as for deriving Theorem~\ref{Theo:Coverage:Typical} and using Lemmas~\ref{Lem:densePBmodel} and \ref{Lem:SignalDense}, the coverage for the current network model is characterized as shown in the following theorem.

\begin{theorem}[Coverage for the  Network with Dense PBs] \label{Coro:denseSP_new} \emph{The success probability, denoted as $P_s'$,  satisfies
\begin{equation}
P_s' \geq \l(1-p_0^{\text{ub}}\r) \mathcal{C}^{\text{lb}}\l(\frac{\theta (1 -\beta D)}{\beta P_c}\r),
\end{equation}
where $\mathcal{C}^{\text{lb}}$ denotes the lower bound on the interference characteristic functional  in Lemma~\ref{Lem:densePBmodel} and $p_0^{\text{ub}}$ is the  upper bound on the  power-outage probability in  Lemma~\ref{Lem:SignalDense}.
}
\end{theorem}

\subsection{Extension to the Network with Dense Micro-PBs}
In this section, we extend the network coverage analysis for the case of dense PBs to the extreme case where the number of cheap micro-PBs per cluster is infinite under a constraint on the sum PB-transmission power per cluster, denoted as $P_{\text{sum}}$. This scenario is of practical interest since the WPT efficiency can be improved by deploying dense small PBs without causing large total power consumption. As a result of the law of large numbers, the transmission power of each backscatter node is stabilized, simplifying the analysis. The network model for this extreme case is modified from the network model with dense PBs by setting the transmission power of each PB in an arbitrary cluster $\mathcal{N}$  to be $P_{\text{sum}}/|\mathcal{N}|$. The power diminishes as the number of PBs in the cluster increases, giving the name ``dense micro-PBs''. Furthermore, the \emph{truncated} path loss model \cite{Baccelli:AlohaProtocolMultihopMANET} is used for WPT links to avoid infinite receive power at nodes due to singularity in the model without truncation. Specifically, given a constant $0< \nu < a$ and transmission power $P$ at a PB at $Y$, the resultant receive power at a target node at $X$ is given as $P[\max(\nu, |X- Y|)]^{-\alpha_1}$.

The sum power received at the typical transmitting node is
\begin{equation}\label{Eq:RX:Pwr}
P_{X_0}'' = \eta g P_{\text{sum}} \lim_{\bar{m}\rightarrow\infty}\frac{1}{|\mathcal{N}_{X_0}|} \sum_{Y\in \mathcal{N}_{X_0}} [\max(\nu, |X_0- Y|)]^{-\alpha_1}.
\end{equation}
Since  $|\mathcal{N}_{X_0}|$ is a Poisson r.v. with mean $\bar{m}$, it is well known that as $\bar{m}\rightarrow\infty$, $|\mathcal{N}_{X_0}|\rightarrow\infty$ almost surely (a.s.). Using this fact as well as that the terms in the sum in \eqref{Eq:RX:Pwr} are i.i.d., it follows from the law of large numbers that
\begin{align}
P_{X_0}'' = 2 \pi \eta g P_{\text{sum}} \int^{\infty}_{0} [\max(\nu, r)]^{-\alpha_1} r f(r) dr, \qquad \qquad \textrm{a.s.}.
\end{align}
By substituting the PDFs in \eqref{Eq:PDF:Matern} and \eqref{Eq:PDF:Thomas}, it is straightforward to obtain the following result.

\begin{lemma}\emph{For the network with dense micro-PBs, as $\bar{m}\rightarrow\infty$, the power received at the typical transmitting node converges as follows.
\begin{itemize}
\item[--] (Matern cluster process)
\begin{align}
P_{X_0}'' = \frac{\eta g P_{\text{sum}}}{a^2} \Big( \frac{\pi}{\alpha_{1}-2} (\nu^{2-\alpha_{1}} - a^{2-\alpha_{1}}) + \nu^{2-\alpha_{1}} \Big), \qquad \qquad \textrm{a.s.},\label{Eq:denseMaternPower}
\end{align}
\item[--] (Thomas cluster process)
\begin{align}
P_{X_0}'' = \eta g P_{\text{sum}} \Big( \nu^{-\alpha_{1}}(1-e^{-\nu}) + (\sqrt{2}\sigma)^{-\alpha_{1}} \Gamma\Big(1-\frac{\alpha_{1}}{2}, \frac{\nu^{2}}{2\sigma^{2}}\Big) \Big), \qquad \qquad \textrm{a.s.}. \label{Eq:denseThomasPower}
\end{align}
\end{itemize}
}
\end{lemma}

Combining the results and the circuit-power constraint leads to the following lemma.

\begin{lemma}[Conditions for Backscatter Communication with Dense Micro-PBs]\label{Lem:BCConditions} \emph{Under the circuit-power constraint, enabling backscatter transmission requires the sum PB-power per cluster to satisfy the following conditions:
\begin{itemize}
\item[--] (Matern cluster process)
\begin{equation}
P_{\text{sum}} \geq \frac{P_c a^2}{\eta g (1-\beta D)\Big( \frac{\pi}{\alpha_{1}-2} (\nu^{2-\alpha_{1}} - a^{2-\alpha_{1}}) + \nu^{2-\alpha_{1}} \Big)},
\end{equation}
\item[--] (Thomas cluster process)
\begin{equation}
P_{\text{sum}} \geq \frac{P_c}{\eta g (1-\beta D) \Big( \nu^{-\alpha_{1}}(1-e^{-\nu}) + (\sqrt{2}\sigma)^{-\alpha_{1}} \Gamma\Big(1-\frac{\alpha_{1}}{2}, \frac{\nu^{2}}{2\sigma^{2}}\Big) \Big)}.
\end{equation}
\end{itemize}
}
\end{lemma}
It can be observed that the sum PB-transmission power for turning on transmitting nodes   scales linearly with the circuit power consumption and decreases  with reducing duty cycle and reflection coefficient as $\frac{1}{1 -\beta D}$.

Last, given constant node-transmission power, the WP-BackCom network is identical to the conventional Poisson distributed MANET (see e.g., \cite{weber2010overview}). If the conditions in Lemma~\ref{Lem:BCConditions} are satisified, the success probability is given as
\begin{equation}\label{Eq:Ps:MicroPB}
P_s'' = \exp\Big( - \frac{2\pi D}{\alpha_{2}}\lambda_{\text{nd}} \theta^{\frac{2}{\alpha_{2}}} \mathcal{B}\Big( \frac{2}{\alpha_{2}}, 1-\frac{2}{\alpha_{2}} \Big)\Big),
\end{equation}
where $\mathcal{B}(x,y)$ is the beta function defined by $\mathcal{B}(x,y) := \int^{1}_{0} t^{x-1}(1-t)^{y-1} dt $.
Increasing the duty cycle densifies backscatter links and thereby strengthens their mutual interference. It is worth noting that the location of typical receiver $z$ does not affect the final result under this model due to the stationary of PPP formed by transmitting nodes. This results in the exponential decay of the success probability with increasing $D$ as observed from \eqref{Eq:Ps:MicroPB}.

\section{Network Capacity}\label{Section:Capacity}

Direct maximization of the network capacity is intractable as the success probability analyzed in the preceding section has no closed form. In this section, for tractability, we consider a WP-BackCom network with \emph{almost-full} network coverage such that transmitted data is always successfully received almost surely. Using \eqref{Eq:Ps:Decomp}, the successful probability can be approximated as $P_s \approx  1 - p_0$. The network capacity is analyzed and optimized for the normal network model. The analysis and discussion can be extended to the network with dense PBs straightforward, which is thus omitted for brevity.

Given $P_s\approx 1 - p_0$, the transmission capacity defined in \eqref{Eq:TXCap} reduces to the density of backscatter nodes:
\begin{equation}
C \approx \lambda_{\text{pb}} \bar{c} D (1-p_0). \label{Eq:Cap:Typical}
\end{equation}
Substituting the results in Lemma~\ref{Lem:SigDist} into \eqref{Eq:Cap:Typical} gives the following:
\begin{itemize}
\item[--] (Matern cluster process)
\begin{equation}
C \approx  \frac{\lambda_{\text{pb}} \bar{c} D}{a^2}\l[\frac{\eta g (1 - \beta D) }{P_c}\r]^{\frac{2}{\alpha_1}},  \label{Eq:Cap:Typical:Matern}
\end{equation}
\item[--] (Thomas cluster process)
\begin{equation}
C \approx  \lambda_{\text{pb}} \bar{c} D \l[1 - \exp\l( -\frac{1}{2\sigma^2}\l(\frac{\eta g (1-\beta D)  }{P_c}\r)^{\frac{2}{\alpha_1}}\r)\r]. \label{Eq:Cap:Typical:Thomas}
\end{equation}
\end{itemize}
First of all, the network capacity is observed to be proportional to the density of backscatter nodes that is consistent with intuition.

\begin{remark}[Effect of $D$ on Network Capacity]\label{Re:Duty} \emph{
Increasing the duty cycle $D$ has two opposite effects on the network capacity, namely increasing the backscatter-node density but reducing transmission probability due to less harvested energy. Therefore, the capacity can be optimized over $D$. For  the model based on the Matern cluster process, the maximum capacity is
\begin{equation}
\max_D C(D) = \frac{\lambda_{\text{pb}} \bar{c}\alpha_1}{a^2(2 + \alpha_1 \beta)}\l[\frac{2 \eta g }{P_c (2 + \alpha_1\beta)}\r]^{\frac{2}{\alpha_1}}, \label{Eq:optimalCapacity}
\end{equation}
and the optimal duty cycle is given as $D^* = \min\l(1, \frac{\alpha_1}{2+\alpha_1 \beta}\r)$. This assumes that $D^*$  is within the constrained range  discussed in Remark~\ref{Re:Cap}.  The capacity optimization for the case of Thomas cluster process is similar but more tedious.}
\end{remark}

\begin{remark}[Effect of $\beta$ on Network Capacity] \label{Re:Reflect} \emph{The effect of $\beta$  on network capacity is  not entirely the same as $D$.  In the  regime of almost-full network coverage,  $C$ is a monotone decreasing function of $\beta$ as observed from \eqref{Eq:Cap:Typical:Matern} and \eqref{Eq:Cap:Typical:Thomas}.
The reason is that  a large reflection coefficient leads to less harvested energy and thereby reduces the backscatter-node density.
 }
\end{remark}

\begin{remark}[Constraints on $D$ and $\beta$] \label{Re:Cap}\emph{It is clear from Remark~\ref{Re:Coverage} that the consideration of the  operational regime of almost-full coverage  constraints $D$ and $\beta$ to be in certain ranges to ensure link reliability. The capacity results in this and next sub-sections hold only for the parameters falling in these ranges.  The corresponding region for $(\beta, D)$ can be derived by bounding the conditional  probability in \eqref{Eq:Ps:Decomp} by a positive value close to one. For instance, using Theorem~$1$, an inner bound of  the region can be derived as
\begin{equation}
\l\{(D, \beta) \in [0, 1]^2  \mid \mathcal{C}\l(\frac{\theta (1 -\beta D)}{\beta P_c}\r) \geq 1 - \epsilon \r\},
\end{equation}
where the positive constant $\epsilon \approx 0$ and $\mathcal{C}$ is the characteristic functional of interference.}
\end{remark}

\section{Simulation Results}

The parameters for the simulation are set as follows unless stated otherwise. The PB transmission power $\eta$ = 40 dBm (10 W) and circuit power is  $P_{c}$ = 7 dBm. The SIR threshold is set as $\theta = -5$ dB in the typical range for ensuring almost-full network coverage (see e.g., \cite{28andrews2011tractable}). The path loss exponents for WPT and backscatter communication links are  $\alpha_{1}$ = 3 and $\alpha_{2}$ = 3, respectively. The reflection coefficient $\beta = 0.6$ and duty cycle $D= 0.4$. The PB density is $\lambda_{\text{pb}} = 0.2\  /\textrm{m}^{2}$ and the expected number of nodes (or PBs) in each cluster $\bar{c} = 3$ (or $\bar{m} = 3$ ). The transmission distance for D2D link is set as $1$ m. The network model based on the Thomas cluster process is assumed with the parameter $\sigma^{2}$ = 4.

\begin{figure}[t]
\centering
\includegraphics[width=8.5cm]{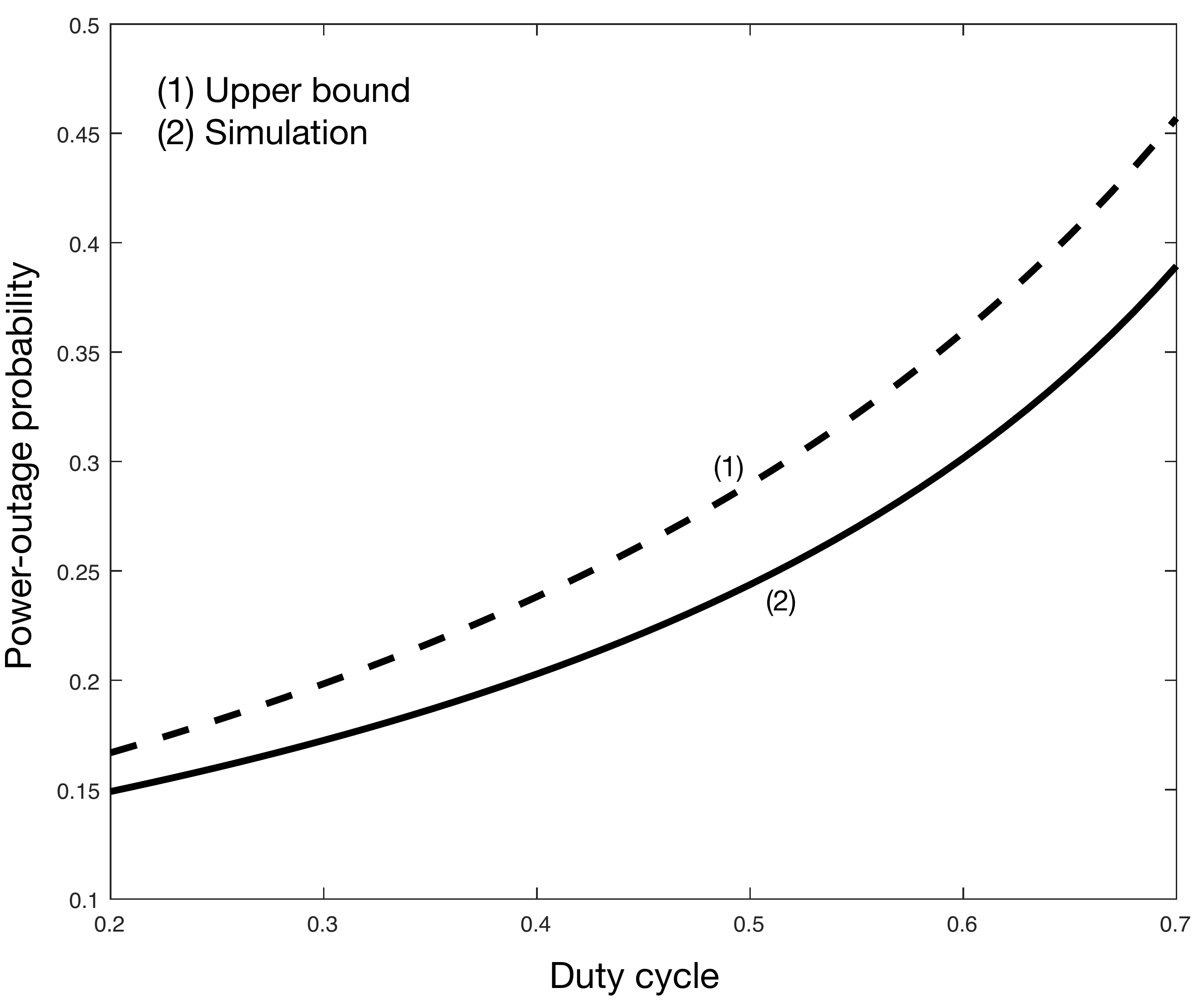}
\caption{Power-outage probability versus duty cycle.}\label{powerOutage}
\end{figure}

The the upper bound on the power-outage probability  in Lemma~\ref{Lem:SignalDense} for the network with dense PBs is evaluated in Fig.~\ref{powerOutage} plotting the probability against the duty cycle $D$. The derived bound is found to be tight. Furthermore, the power-outage probability is observed to grow with increasing $D$ and converge to $0$ as $D$ approaches $0$.  The reason is that large $D$ means that more time is spent on backscattering and less on energy harvesting, reducing the amount of harvested energy for operating the node circuits. The bound is also found to be tight by varying the circuit power consumption and reflection coefficient. The corresponding simulation results are omitted for brevity.

\begin{figure}[t]
\centering
\includegraphics[width=8.5cm]{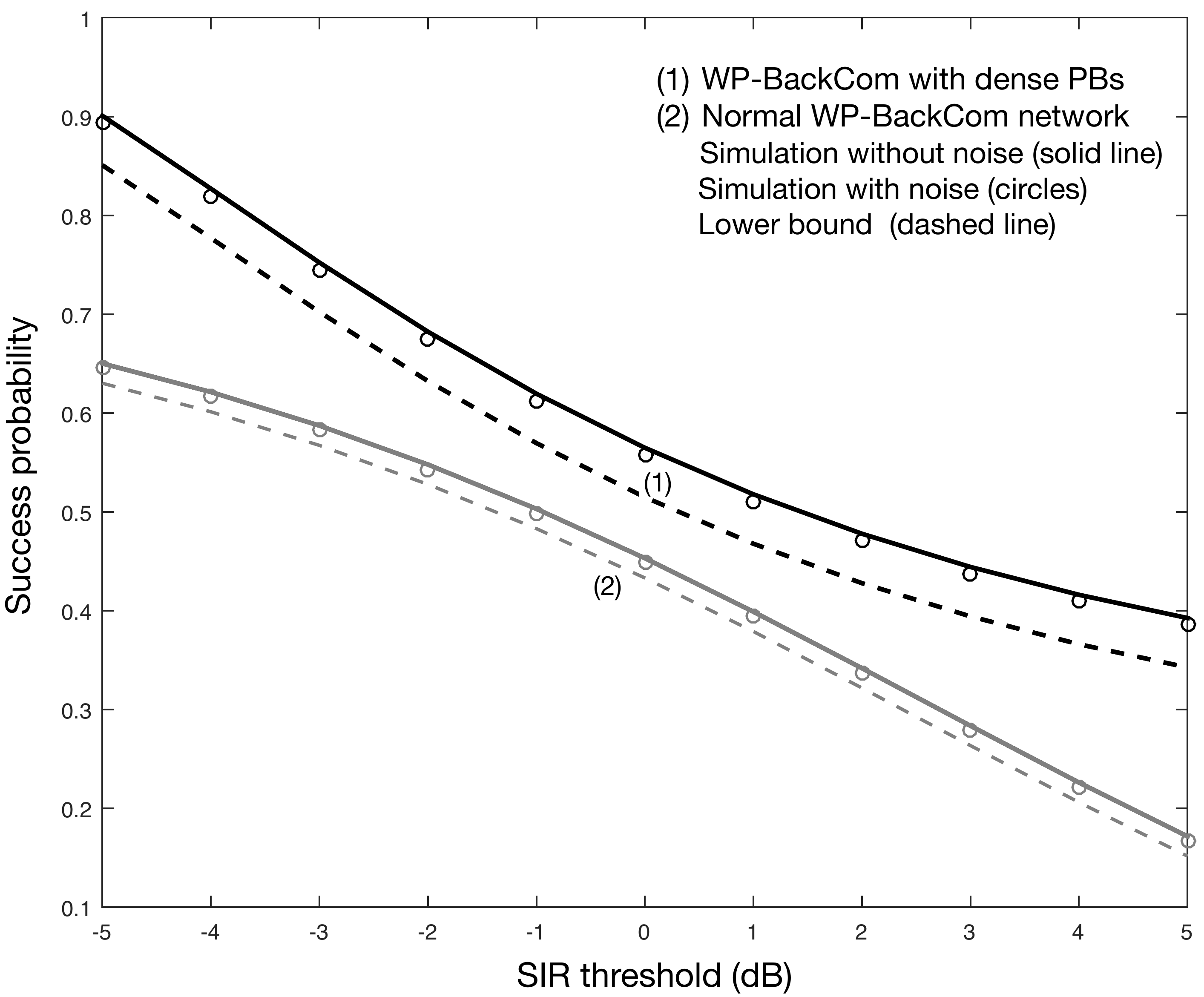}
\caption{Success probability versus SIR threshold for both the normal WP-BackCom network and the network with dense PBs (with and without thermal noise). For the case with noise, the noise variance is set as $-90$ dBm. }\label{threeModels}
\end{figure}

In Fig.~\ref{threeModels}, the success probability for backscatter communication is plotted against the SIR threshold $\theta$ (in dB) for both the normal WP-BackCom network and the network with dense PBs with and without considering the impact of thermal noise. The derived theoretic lower bound is also plotted for comparison. One can see that deploying dense PBs substantially increases the success probability by reducing the likelihood of power outage. Moreover, the theoretic bounds on the success probability are observed to be tight for both networks. Fig.~\ref{threeModels} also illustrates the effects of the thermal noise on success probability. One can  see that the simulation results with (circle points) and without noise (solid lines) are close to each other, which validates the assumption of interference limited network.

\begin{figure}[t]
\centering
\subfigure[Effect of the reflection coefficient]{\includegraphics[width=8cm]{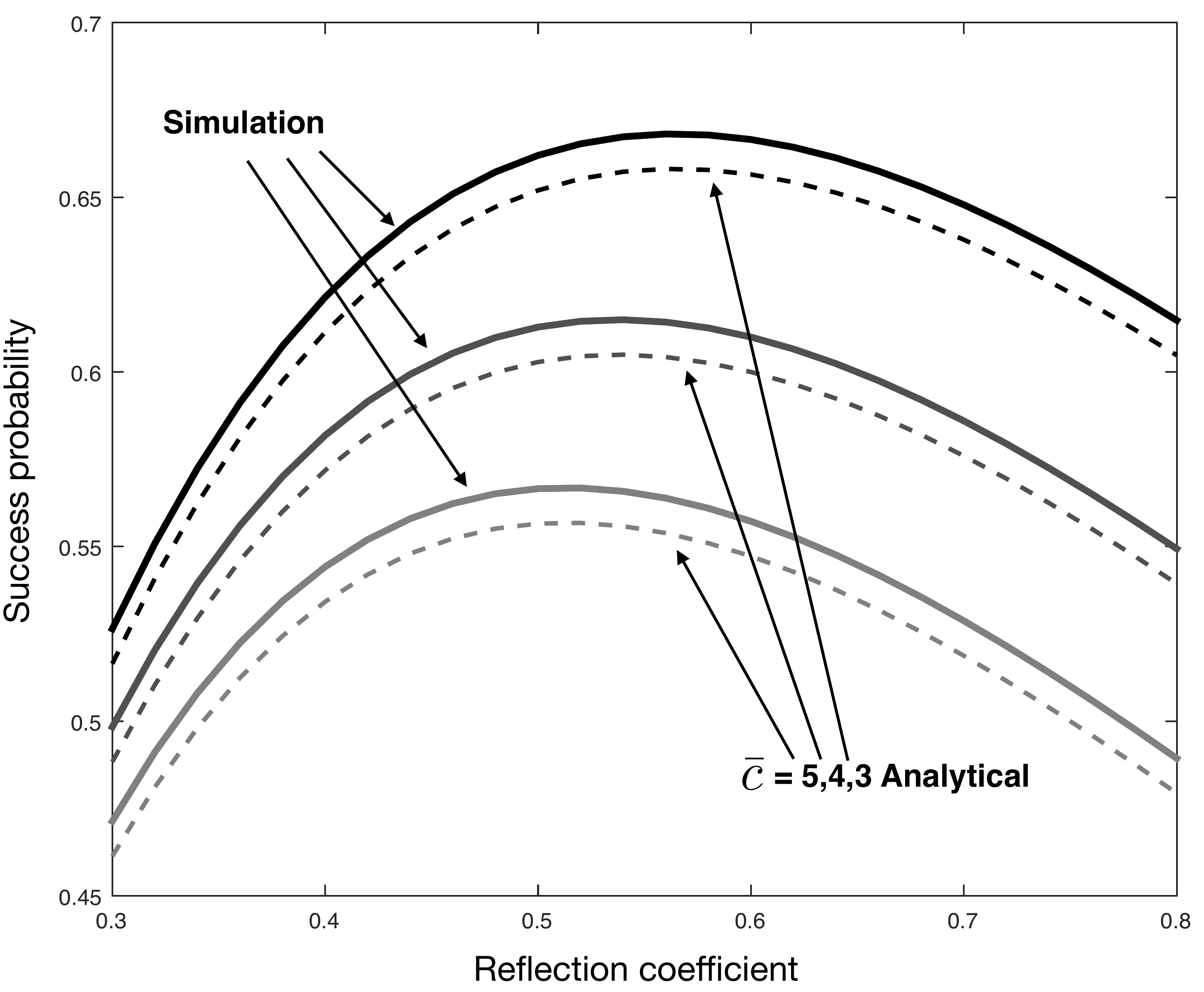}}
\subfigure[Effect of the duty cycle]{\includegraphics[width=8cm]{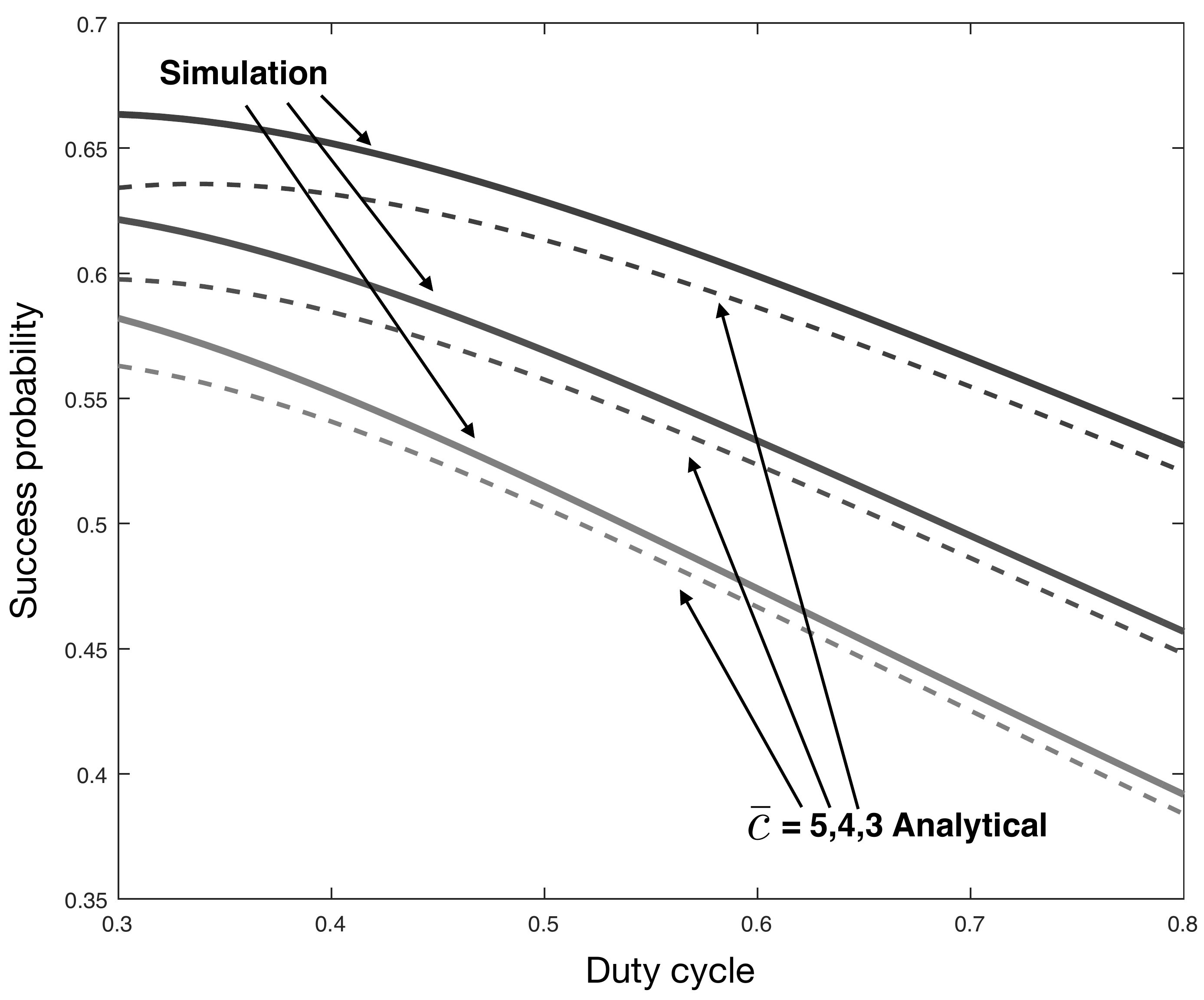}}
\caption{The effects of the backscatter parameters, namely the reflection coefficient and duty cycle, on the success probability for the normal WP-BackCom network with a variable expected number of backscatter nodes per cluster,  $\bar{c} \in \{3,4,5\}$. }\label{cpSim}
\end{figure}

The curves of success probability versus the duty cycle and reflection coefficient are plotted in  Fig.~\ref{cpSim} for different values of $\bar{c}$. The curves based on the analytical results in Theorem~\ref{Theo:Coverage:Typical} are plotted for comparison. It is observed that the theoretical lower  bounds are tight. The curves show that the success probability is the concave function of the reflection coefficient, which is consistent with the discussion in Remark~\ref{Re:Coverage}. The optimal value for the reflection coefficient is observed to be about $0.5 - 0.6$. Consider duty cycle $D$ in Fig.~\ref{cpSim}(b), the success probability reduces with increasing $D$. It is intuitive to understand because larger $D$ results in higher density of interference and larger power-outage probability. Furthermore, the lower bounds  are not very tight when the value of $D$ is small. The reason is that smaller $D$ leads to higher receive power and transmission probability at receiver, using $\frac{\beta P_c}{1 - \beta D}$ to replace $P_t$ (for obtaining the lower bound) incurs a little larger gap.

\begin{figure}
\centering
\includegraphics[width=8.5cm]{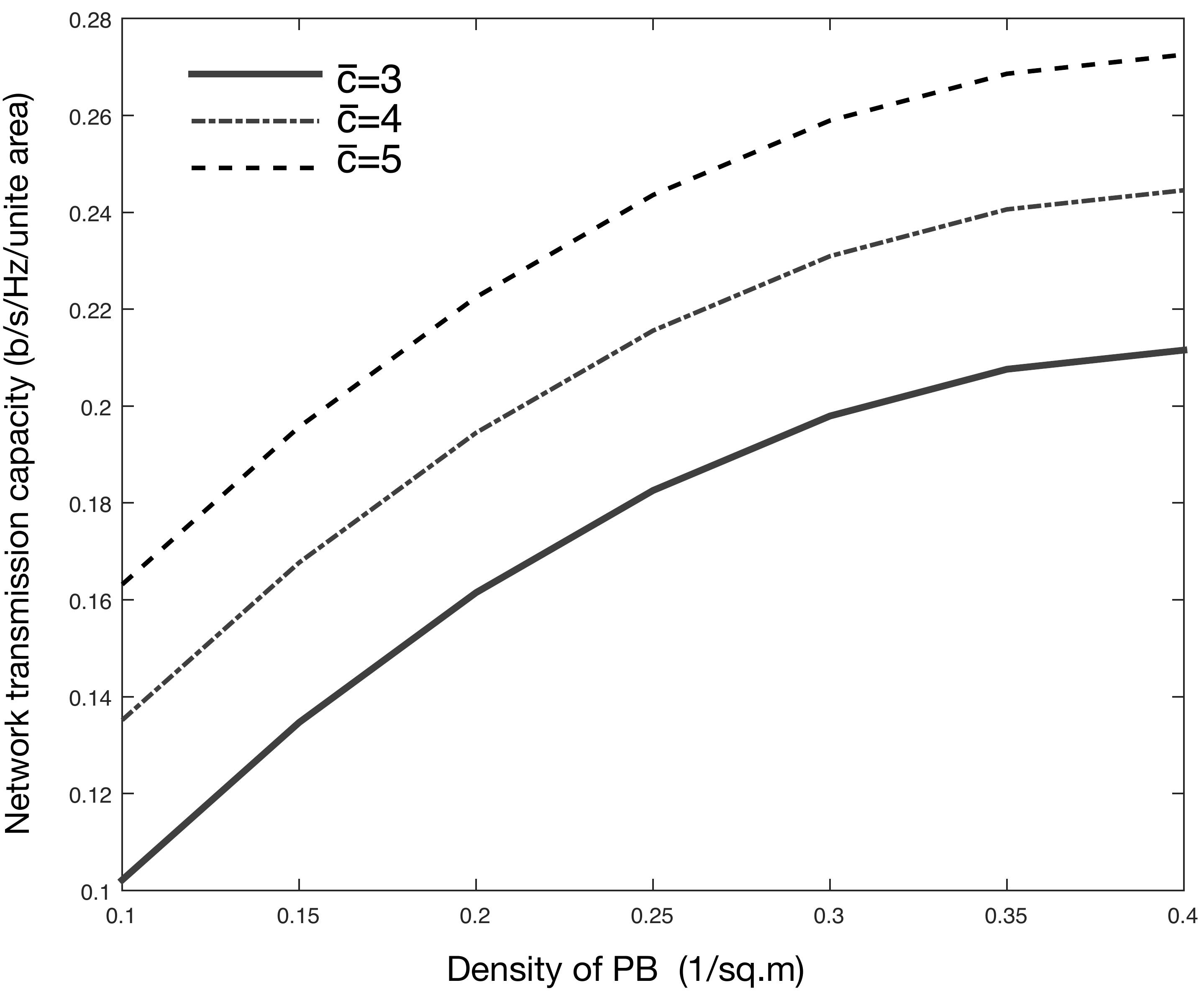}
\caption{The network transmission capacity versus the PB density for the normal WP-BackCom network with a variable expected number of backscatter nodes per cluster,  $\bar{c} \in \{3,4,5\}$}
\label{Fig:TxCap:A}
\end{figure}

The curves of network transmission capacity  versus the PB density are shown in Fig.~\ref{Fig:TxCap:A} for different values of $\bar{c}$.
When the density of PB is relatively small (e.g., $\lambda_{\text{pb}}\in [0.1, 0.25]$) , the network capacity is observed to grow linearly with the PB density. This is aligned with theoretic result in Section~\ref{Section:Capacity} for the regime of almost-full network coverage.  For a large PB density, the capacity saturates as the network becomes dense and interference limited. The linear relation no longer holds since the success probability no longer be approximated as $1$. In addition, it can be observed from Fig.~\ref{Fig:TxCap:A} where increasing  $\bar{c}$ contributes an approximately constant capacity gain insensitive to changes on the PB density.

Last, the effects of the backscatter parameters, namely the reflection coefficient $\beta$ and duty cycle $D$, on the network transmission capacity are shown in Fig.~\ref{Fig:TxCap:B} for different values of $\bar{c}$. From Fig.~\ref{Fig:TxCap:B}(a), it can be observed that the  capacity is a convex function of the duty cycle,  which confirms Remark~\ref{Re:Duty}. The optimal duty cycle is almost identical (about $0.62$) for different values of $\bar{c}$. Furthermore, the capacity grows with the increasing  node density and the gain is the largest at $D\approx 0.62$. Consider the curves of capacity versus reflection coefficient in Fig.~\ref{Fig:TxCap:B}(b). For $\beta$ larger than $0.3$ core, the capacity decreases rapidly as $\beta$ increases due to growing power-outage probability. This range of $\beta$ corresponds to almost-full network coverage and the observation is consistent with  Remark~\ref{Re:Cap}.

\begin{figure}
\centering
\subfigure[Effect of the duty cycle given reflection coefficient $\beta = 0.6$]{\includegraphics[width=8cm]{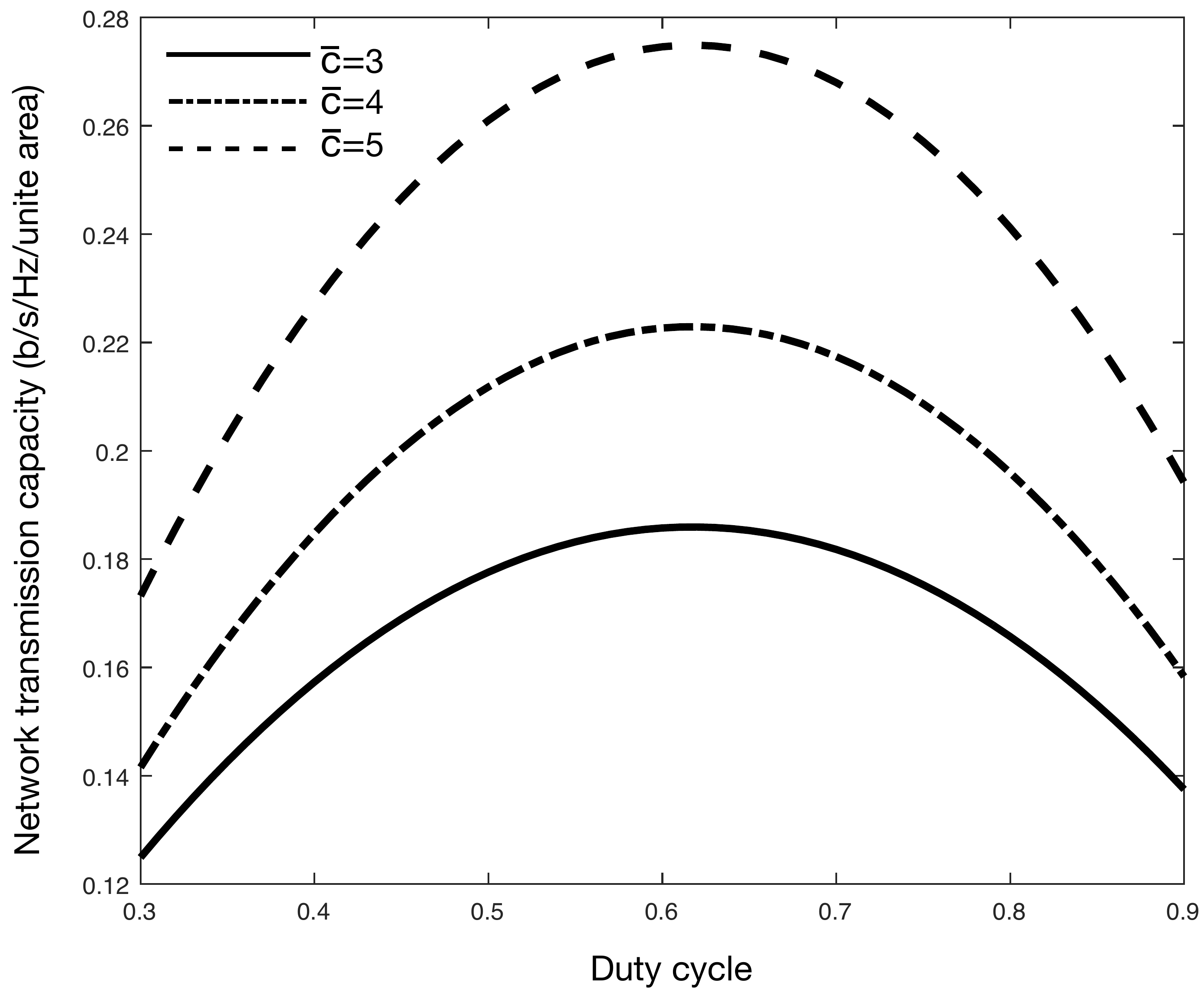}}
\subfigure[Effect of the reflection coefficient given duty cycle $D = 0.4$]{\includegraphics[width=8cm]{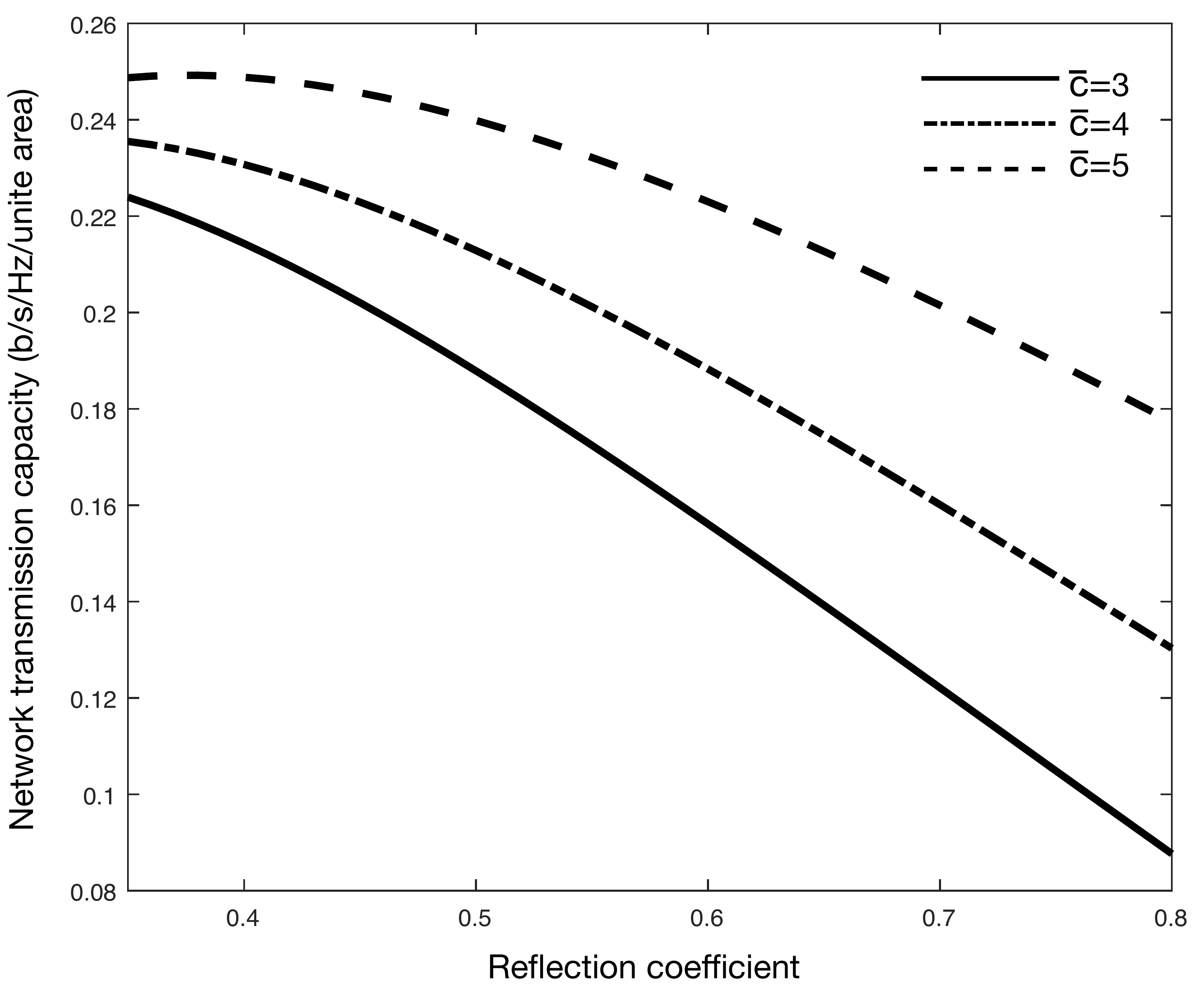}}
\caption{The effects of the duty cycle and reflection coefficient on the network transmission capacity for a variable expected number of backscatter nodes per cluster,  $\bar{c} \in \{3,4,5\}$.}
\label{Fig:TxCap:B}
\end{figure}

\section{Conclusion}
In this work, we have proposed the new network architecture, namely the WP-BackCom network,  for realizing dense backscatter communication networks using WPT enabled by PBs. A large-scale WP-BackCom network has been modeled using the PCP. Applying  stochastic geometry theory, the success probability and the network transmission capacity have  been derived to quantify the performance of network coverage and capacity, respectively. In particular, the results relate the network performance to  the backscatter parameters, namely the duty cycle and the reflection coefficient.

The current work can be extended in several directions. In this work, each PB serves only a single node assuming beamforming. However, PB with isotropic transmission can serve multiple nodes, which makes the network-performance analysis more challenging but can lead to new insight into WP-BackCom network design. Next, following a similar approach as in this work,  other types of backscatter communication networks, such as those based on UWB or ambient RF energy harvesting, can be modeled and studied using stochastic geometry comprising a wide range of spatial point processes  including cluster point processes other than PCP, e.g., Gauss poisson processes and repulsive point processes (Ginibre point processes).  Last, re-generated interference arising from backscattering interference by nodes is omitted in the current work. It is interesting and also important to investigate the effect of interference regeneration on the network performance.

\appendix

\subsection{Proof of Lemma \ref{Lem:IntraInt}} \label{appendix:lemma:1}
Let  $\E_{\Phi}$  denote the expectation with respect to a random process/variable $\Phi$. Consider the typical cluster of the transmitting node process. Given  the typical receiver  located  at $z\in\mathds{R}^2$, the characteristic functional of the intra-cluster interference in \eqref{Eq:IntraInt} is obtained as
\begin{align}
\E\left[e^{-s I_a}\right]  &= \E_{\mathcal{N}_{0}\backslash \{X_0\}, Y_0}\Bigg[ \exp\Bigg(-s \sum\limits_{X\in \mathcal{N}_{0}\backslash \{X_0\} }\beta h_X \ell(|X-Y_0|^{-\alpha_{1}}) |X\!-\!z|^{-\alpha_{2}}\Bigg) \Bigg] \nonumber \\
&\overset{(a)}{=}  \E_{\mathcal{N}_{0}, Y_0}\Bigg[ \exp\Bigg(-s \sum\limits_{X\in \mathcal{N}_{0}}\beta h_X \ell(|X-Y_0|^{-\alpha_{1}}) |X\!-\!z|^{-\alpha_{2}}\Bigg) \Bigg]
 \nonumber \\
& \overset{(b)}{=} \E_{Y_0}\Bigg[ \exp\Bigg(-\bar{c} D \int\limits_{\mathds{R}^2}\Big( 1 - \E_h\Big(\exp\Big(- s \beta h \ell(|x-Y_0|^{-\alpha_{1}})|x-z|^{-\alpha_{2}}\Big)\Big) \times \nn\\
& \qquad \qquad \qquad f(|x - Y_0|)d x\Bigg)\Bigg]\nn\\
& = \E_{Y_0}\Bigg[ \exp\Bigg(\!\!-\bar{c} D \int\limits_{\mathds{R}^2}\frac{1}{1 + \frac{1}{s \beta \ell(|x-Y_0|^{-\alpha_{1}})|x-z|^{-\alpha_{2}}}} f(|x - Y_0|)d x\Bigg)\Bigg]\nn\\
&=\E_{Y_0}\Bigg[ \exp\Bigg(\!\!\!-\bar{c} D \int\limits_{\mathds{R}^2}\frac{1}{1 + \frac{1}{s \beta \ell(|x|^{-\alpha_{1}})|x + Y_0 -z|^{-\alpha_{2}}}}f(|x|)d x\Bigg)\Bigg], \label{appendixA}
\end{align}
where $(a)$ and $(b)$ are obtained by applying  Slivnyak's Theorem an Campbell's Theorem (see e.g., Chapter~4.5 in \cite{33haenggi2012stochasticgeometry}), respectively.  Using the PDF $f(\cdot)$  in \eqref{Eq:PDF:Matern} and \eqref{Eq:PDF:Thomas} and the circuit-constraint function $\ell(\cdot)$ defined before \eqref{Eq:IntPwr:a}, the desired result follows.

\subsection{Proof of Lemma \ref{Lem:InterInt}} \label{appendix:lemma:2}
By applying Slivnyak's Theorem, the characteristic functional of the inter-cluster interference $I_b$ in \eqref{Eq:InterInt} is obtained as
\begin{align}
\E\left[e^{-s I_b}\right] &= \E\Bigg[ \exp\Bigg(-s \sum_{Y\in \Pi_{\text{pb}} }\sum_{X\in \mathcal{N}_{Y} } \beta h_X\ell(|X - Y|^{-\alpha_1})|X-z|^{-\alpha_2} \Bigg) \Bigg] \nonumber \\
&=\E\Bigg[ \prod_{Y\in \Pi_{\text{pb}}}\E\Bigg[\exp\Big(-s \sum_{X\in \mathcal{N}_{Y} } \beta  h_X \ell(|X - Y|^{-\alpha_1})|X-z|^{-\alpha_2} \Big) \Bigg]\Bigg] \nonumber.
\end{align}
The inner expectation focusing on a single cluster  can be derived using similar steps as in the proof of Lemma~\ref{Lem:IntraInt}. As a result,
\begin{align}
\E\left[e^{-s I_b}\right] &= \E\Bigg[  \exp\Bigg(-\sum_{Y\in \Pi_{\text{pb}}}\bar{c} D q(s, Y, z) \Bigg)\Bigg].
\end{align}
where $q(\cdot)$ is defined in Lemma~\ref{Lem:IntraInt}.  Applying Campbell's Theorem gives the desired result.

\subsection{Proof of Lemma \ref{Lem:densePBmodel}} \label{appendix:lemma:5}
Using  the  definition of $I^{'}$ in (\ref{De:InterDense}) and applying Slivnyak's Theorem,
\begin{align}
\E[e^{-sI^{'}}] &= \E\Big[ \exp\Big( -s \beta \sum_{X \in  \Pi_{\text{nd}}\backslash \{X_{0}\}} \Big( \sum_{Y \in \mathcal{N}_{X}} \eta g |Y-X|^{-\alpha_{1}} \Big) h_{X} |X-z|^{-\alpha_{2}} \Big) \Big] \nonumber \\
&= \E\Big[ \prod_{X \in \Pi_{\text{nd}}\backslash \{X_{0}\}} \exp\Big( -s \beta \Big( \sum_{Y \in \mathcal{N}_{X}} \eta g |Y-X|^{-\alpha_{1}} \Big) h_{X} |X-z|^{-\alpha_{2}} \Big) \Big] \nonumber \\
&\overset{(a)}{=}  \exp\Big( -\lambda_{\text{nd}}D \int_{\mathds{R}^2} \Big( 1- \E_{\mathcal{N}_{X}}\E_{h}\Big[ \exp\Big( -s \beta \Big( \sum_{Y \in \mathcal{N}_{X}} \eta g |Y-x|^{-\alpha_{1}} \Big) h |x-z|^{-\alpha_{2}} \Big)\Big] \Big)dx\Big) \nonumber \\
&\overset{(b)}{=} \exp\Big( -\lambda_{\text{nd}}D \int_{\mathds{R}^2} \Big( 1- \nonumber \\
& ~~~~~~~~\exp\Big( -2\pi\bar{m} \int_{0}^{\infty} \Big(1-\E_{h} \Big[ \exp\Big( -s \beta \eta g h |x-z|^{-\alpha_{2}} r^{-\alpha_{1}} \Big)\Big] \Big) f(r)rdr \Big) \Big)dx\Big) \nonumber \\
&= \exp\Big( -\lambda_{\text{nd}}D \int_{\mathds{R}^2} \Big( 1- \exp\Big( -2\pi\bar{m} \int_{0}^{\infty} \frac{1}{1+(s \beta \eta g)^{-1} r^{\alpha_{1}} |x-z|^{\alpha_{2}}}f(r)rdr \Big) \Big)dx\Big)
\end{align}
where both (a) and (b) are obtained by applying Campbell's Theorem. The desired result follows.

\subsection{Proof of Lemma \ref{Lem:SignalDense}} \label{appendix:proposition:2}
For the bound on the power-outage probability  in   \eqref{Eq:Chernoff}, the term $\sum^{N}_{n} d_{n}^{-\alpha_{1}}$ is a compound Poisson process (or equivalently  a shot-noise process) . Then using the result in   \eqref{Eq:Chernoff} and the characteristic functional of a shot-noise process (see e.g., \cite{HaenggiAndrews:StochasticGeometryRandomGraphWirelessNetworks}),
\begin{align}
p_0'  &= \min_{\mu >0}\Big\{ \exp\Big( \frac{\mu P_c}{(1-\beta D) \eta g} - 2\pi \bar{m} \int^{\infty}_{0} \Big( 1- e^{-\mu t^{-\alpha_{1}}}\Big)f(t)t dt\Big) \Big \}.
\end{align}
Note that the function of $\mu$to be minimized is convex. Thus the optimal value of $\mu$, denoted as $\mu^*$, can be found by solving the equation from setting the derivative of the  said function as zero, yielding the result in the lemma statement.

\bibliographystyle{ieeetr}

\end{document}

%% file: header.tex
\newtheorem{theorem}{Theorem}

\newtheorem{lemma}{Lemma}

\newtheorem{remark}{\bf Remark}

\def\E{\mathsf{E}}

\def\phi{\varphi}

\def\l{\left}
\def\r{\right}
\def\({\left(}
\def\){\right)}

\def\b0{{\mathbf{0}}}

\newcommand{\nn}{\nonumber}